\documentclass[11pt]{article}
\usepackage{amssymb,amsmath,amsthm,mathtools}
\usepackage{graphicx,fullpage}
\usepackage{verbatim}
\usepackage{float}

\usepackage{enumerate}

\usepackage{multirow}

\usepackage{multicol}

\usepackage{tikz}

\usepackage{caption}
\usepackage{dsfont} 

\newtheorem{theorem}{Theorem}
\newtheorem{thm}[theorem]{Theorem}
\newtheorem{lemma}[theorem]{Lemma}
\newtheorem{cor}[theorem]{Corollary}

\theoremstyle{definition}
\newtheorem{example} {Example}

\usepackage{changepage}

\newenvironment{changemargin}[2]{%
\begin{list}{}{%
\setlength{\topsep}{0pt}%
\setlength{\leftmargin}{#1}%
\setlength{\rightmargin}{#2}%
\setlength{\listparindent}{\parindent}%
\setlength{\itemindent}{\parindent}%
\setlength{\parsep}{\parskip}%
}%
\item[]}{\end{list}}

\def\M{{\mathcal M}}
\def\P{{\mathcal P}}
\def\Q{{\mathcal Q}}
\def\R{{\mathcal R}}
\def\S{{\mathcal S}}
\def\ext{\text{2-}ext}
\newcommand{\ie} [1] {\textit{i.e.,} #1}
\usepackage{mathtools}

\begin{document}

\title{Constructing Permutation Arrays using Partition and Extension}
\author{Sergey Bereg\thanks{
Department of Computer Science,
University of Texas at Dallas,
Box 830688,
Richardson, TX 75083
USA. {\quad }
Sergey Bereg was supported in part by NSF award CCF-1718994.}
\and Luis Gerardo Mojica$^*$
\and Linda Morales$^*$
\and Hal Sudborough$^*$
}

\date{}

\maketitle

\begin{abstract}
We give new lower bounds for $M(n,d)$, for various positive integers $n$ and $d$ with $n>d$, where $M(n,d)$ is the largest number of permutations on $n$ symbols with pairwise Hamming distance at least $d$. 
Large sets of permutations on $n$ symbols with pairwise Hamming distance $d$ are needed for constructing error correcting permutation codes, which have been proposed for power-line communications. 
Our technique, {\em partition and extension}, is universally applicable to constructing such sets for all $n$ and all $d$, $d<n$.
We describe three new techniques, {\em sequential partition and extension},  {\em parallel partition and extension}, and a {\em modified Kronecker product operation}, which extend the applicability of partition and extension in different ways. 
We describe how partition and extension gives improved lower bounds for $M(n,n-1)$ using mutually orthogonal Latin squares (MOLS).
We present efficient algorithms for computing new partitions: an iterative greedy algorithm and an algorithm based on integer linear programming. 
These algorithms yield partitions of positions (or symbols) used as input to our partition and extension techniques. 
We report many new lower bounds for $M(n,d)$ found using these techniques for $n$ up to $600$. 
\end{abstract}

\section{Introduction}
\label{s:intro}

The use of permutation codes for error correction of communications transmitted over power-lines has been suggested \cite{huczynska2006powerline,pavlidou2003power}. Due to the extreme noise in such channels, codewords are sent by frequency modulation rather than by amplitude modulation. Let's say we use frequencies $f_{0}, f_{1}, f_{2}, \dots, f_{n-1}$, which we view by the index set $Z_n =\{0,1,2,\cdots,n-1\}$. A permutation on $Z_n$, corresponding to a codeword, specifies in which order frequencies are to be sent. 

The Hamming distance between two permutations, $\sigma$ and $\tau$ on $Z_n$, denoted by $hd(\sigma,\tau)$, is the number of positions $x$ in $Z_n$ such that $\sigma(x)\neq \tau(x)$.  For example, the permutations on $Z_5$, $\sigma = 0~4~1~3~2$ and $\tau= 2~4~3~1~2$ have $hd(\sigma,\tau)=3$, as they differ in positions $0, 2,$ and $3$. 
A set $A$ of permutations on $Z_n$ (called a \textit{permutation array} or \textit{PA} for short) has Hamming distance $d$, denoted by $hd(A)\geq d$, if, for all $\sigma,\tau \in A, ~hd(\sigma,\tau) \geq d$. 
The maximum size of a PA $A$ on $Z_n$ with $hd(A) \geq  d$ is denoted by $M(n,d)$. 
Two PAs $A$ and $B$ have Hamming distance $d$, denoted by $hd(A,B)\geq d$, if, for all $\sigma \in A \text{ and } \tau \in B, ~hd(\sigma,\tau) \geq d$.

There are known combinatorial upper and lower bounds on $M(n,d)$, specifically the Gilbert-Varshamov \textit{(GV)} bounds, together with some recent improvements to the \textit{GV} bounds \cite{deza1978bounds,gao2013improvement,wang17}. 
Generally, these bounds are theoretical and are often improved by empirical techniques. 
Some exact values are known: \textit{(1)} for all $n$, $M(n,n)=n$, and, \textit{(2)} for $q$, a power of a prime, $M(q,q-1)=q(q-1)$ and $M(q+1,q-1)=(q+1)q(q-1)$. 
These exact values come from sharply $k-$transitive groups, for $k=2$ and $k=3$, namely the affine general linear group, denoted by \textit{AGL}, and the projective general linear group, denoted by \textit{PGL} \cite{deza1978bounds,conway1985atlas}. 
The Mathieu sharply $4-$transitive and $5-$transitive groups, give exact values for $M(11,8)=7920$ and $M(12,8)=95040$ \cite{cameron1999permutation,conway1985atlas,dixon1996permutation}.
It is not feasible to do an exhaustive search for good permutation arrays when $n$ becomes large. 
There are $n!$ permutations on $Z_n$, so the search space becomes computationally impractical. 
Some researchers have attempted to mitigate the problem by considering automorphisms groups and replacing permutations by sets of permutations. For example, in \cite{jani15}, Janiszczak, et. al. considered sets of permutations invariant under isometries to improve several lower bounds for $M(n,d)$, for various choices of $n$ and $d$, $n\leq22$.  Chu, Colbourn and Dukes \cite{chu2004constructions} and Smith and Montemanni \cite{smith2012new} also provide lower bounds obtained by the use of automorphism groups, and are also generally limited to small values of $n$.

There is also a connection between mutually orthogonal Latin squares (MOLS) and permutation arrays \cite{colbourn2004permutation}. Specifically, if there are $k$ mutually orthogonal Latin squares of side $n$, then $M(n,n-1)\geq kn$. Let $N(n)$ denote the number of mutually orthogonal Latin squares of side $n$. Finding better lower bounds for $N(n)$ is an on-going combinatorial problem of considerable interest world-wide \cite{von1983bemerkung,colbourn2006handbook}. 

Recently, we described a new technique, called {\em partition and extension} \cite{bmms17,bms17} and we illustrated how to use this technique to improve several lower bounds for $M(n,n-1)$ over those given by MOLS. Partition and extension operates on permutation arrays that can be decomposed into subsets with certain properties. (A description follows in Section \ref{s:prev}.) 
In its simplest form, partition and extension converts a PA $A$ on $n$ symbols with $hd(A)=d-1$, into a PA $A'$ on $n+1$ symbols with $hd(A')=d$. 
That is, when a PA $A$ exhibiting $M(n,d-1)$ meets the necessary conditions for simple partition and extension, the technique obtains a lower bound for $M(n+1,d)$. 

The purpose of this paper is to illustrate many new ways to use the partition and extension technique, and ways to generate appropriate partitions. 
We describe a method called {\em sequential partition and extension}, an improvement which uses iteration to extend permutation arrays by two or more symbols. 
When certain conditions are met, sequential partition and extension obtains new PAs on $n+2$ symbols with Hamming distance $d$ from PAs on $n$ symbols with Hamming distance $d-1$.
Another new technique, which we call {\em parallel partition and extension}, introduces several new symbols simultaneously. 
In some cases, parallel partition and extension on PAs on $n$ symbols with Hamming distance $d-r$ gives new lower bounds for $M(n+r,d)$.
We illustrate how to use partition and extension on blocks defined by cosets of the cyclic subgroup of the group $AGL(1,q)$, and on PAs created by a modified Kronecker product operation. 
We give new results derived from
partition and extension on blocks defined by mutually orthogonal Latin squares (MOLS).
We describe experimental algorithms and heuristics for creating partitions, including a greedy algorithm and an optimization approach based on Integer Linear Programming.
These new techniques improve on previously reported results  \cite{bms17}.

\section{Previous Results on Partition and Extension}
\label{s:prev}

We briefly describe the technique called {\em partition and extension}, which transforms a PA on $Z_n$ with Hamming distance $d-1$ into a PA on $Z_{n+1}$ with Hamming distance $d$.
A detailed description and several examples appear in \cite{bms17}. 
Throughout this paper we will use the phrase {\em simple partition and extension} to refer to this version of partition and extension.

Let $s$ be a positive integer. Let $M_1, M_2, \dots, M_s$ be an ordered list of $s$ pairwise disjoint 
permutation arrays on $Z_n$. Let ${\mathcal P}=(P_1, P_2, \dots, P_s )$ and 
${\mathcal Q}=(Q_1, Q_2, \dots, Q_s )$ 
be two ordered lists of subsets of $Z_n$ such that the sets in $\P$ and $\Q$ are partitions of $Z_n$.   
For each set $M_i$, $P_i$ is the set of locations and $Q_i$ is the set of symbols to be replaced by the new symbol $n$. 
When a permutation $\sigma$ in $M_i$ has a symbol $q$ in $Q_i$ appearing in a position $p$ in $P_i$,  $\sigma$ is extended (i.e., converted to a permutation $\sigma'$ on $n+1$ symbols) by moving $q$ to the end of the permutation and placing the symbol $n$ in position $p$. 
That is, the {\em extension of $\sigma$ by position $k$}, denoted by $ext_k(\sigma)=\sigma'$, 
is a permutation on $Z_{n+1}$ defined by: 
$\sigma'(k)=n,\sigma'(n)=\sigma(k)$, and for all $j$ $(0\le j<n,j\ne k)$, $\sigma'(j)=\sigma(j)$. 
We refer to this new permutation as $ext(\sigma)$ and $\sigma'$ interchangeably. 

For each $i$, let $covered(M_i)$ be the subset of $M_i$, defined by $covered(M_i)=\{\sigma \in M_i  ~|~ \exists p\in P_i, ~ \sigma(p)\in Q_i \}$.  
We say that a permutation $\sigma$ is {\em covered} if $\sigma \in covered(M_i)$ for some $i$.  
In order for a permutation $\sigma'$ to be included in the extended set of permutations on $Z_{n+1}$, $\sigma$ must be covered. 
That is, $\sigma$ must have one of the named symbols in one of the named positions. 
In general, when $\sigma\in covered(M_i)$, there may be more than one position $p\in P_i$ 
such that $\sigma(p)\in Q_i$.
If so, arbitrarily designate one of these positions
to cover $\sigma$.   

For our construction, we include an additional PA $M_{s+1}$, 
for which there is no corresponding set of positions or symbols. 
None of the permutations in $M_{s+1}$ are in any of the PAs $M_{i}$.  
The partition and extension operation adds the new symbol $n$ to the end of each permutation in $M_{s+1}$.
Every permutation in $M_{s+1}$ is used in the construction 
of our new PA. 
Thus, we create the list
${\mathcal M}=(M_1, M_2, \dots, M_{s+1})$, which includes this extra set. 

A triple $\Pi=({\mathcal M,P,Q})$ is a {\em distance-$d$ partition system} for $Z_n$ if it 
satisfies the following properties:
\begin{enumerate}[(I)]
\item \label{dist-d prop1} $\forall M_i\in{\mathcal M},~ hd(M_i)\ge d$, and
\item \label{dist-d prop2 } $\forall i,j ~(1\le i<j\le s+1),~ hd(M_i,M_j)\ge d-1$.
\end{enumerate}

Simple partition and extension uses sets $P_i$ and $Q_i$ in the two partitions $\P$
and $\Q$ to modify the covered permutations in $M_i$, for $1\le i\le s$, 
for the purpose of creating a new PA on $Z_{n+1}$ with Hamming distance $d$.  
Let $\Pi=(\M,\P,\Q)$ be a distance-$d$ partition system, where 
${\mathcal M}=(M_1, M_2, \dots, M_{s+1})$, for some $s$.  
We now show how the simple partition and extension operation creates a new permutation array 
$ext(\Pi)$ on $Z_{n+1}$.
For all $i$ $(1\le i\le s)$, let $ext(M_i)$ be the set of permutations defined by 
\[
ext(M_i)=\{ext(\sigma)~|~\sigma \in covered (M_i)\}.
\] 
For $M_{s+1}$, let $ext(M_{s+1})$ be the set of permutations on $Z_{n+1}$ defined by 
adding the symbol $n$ to the end of every permutation of $M_{s+1}$. 

Let $ext(\Pi)$ be the set of permutations on $Z_{n+1}$ defined by 
\[
ext(\Pi)=\bigcup_{i=1}^{s+1} ext(M_i).
\]
Note that 
\begin{equation}
\label{eq:sizeof_extPi}
|ext(\Pi)|=\sum_{i=1}^{s+1} |ext(M_i)|.
\end{equation}

\begin{thm} [{\cite{bms17}}]
\label{th:simple_pe}
Let $d$ be a positive integer. 
Let  $\Pi=(\M,\P,\Q)$ be a distance-$d$ partition system for $Z_n$, 
with $\M=(M_1,M_2,\dots,M_{s+1})$ for some positive integer $s$.
Let $ext(\Pi)$ be the PA on $Z_{n+1}$ created by simple partition and extension.
Then, $hd(ext(\Pi))\ge d$. 
\end{thm}

The example in Table \ref{tb:toy_pe_ex} illustrates the application of Theorem \ref{th:simple_pe} to $\Pi=(\M,\P,\Q)$, where $\M=(M_1,M_2,M_3),~ \P=\{ \{0,2\},\{1,3 \} \}$ and $\Q=\{ \{0,1\},\{2,3 \} \}$. 
The column on the left shows the PAs $M_1$, $M_2$ and $M_2$.
$M_1$ is the cyclic subgroup of $AGL(1,4)$, and $M_2$ and $M_3$ are two of its cosets. 
The blue symbols are the symbols of $Q_i$ that occupy positions in $P_i$, for $i \in {1,2}$.
The column on the right shows the new PAs obtained by simple partition and extension on $\Pi$.
To create $ext(M_1)$ and $ext(M_2)$, the blue symbols are moved to the end of the permutations and a new symbol, 4, in red, occupies the positions vacated by the blue symbols.
To create $ext(M_3)$, the symbol 4 is simply appended to the end of each permutation.
Note that $hd(M_1) \geq 4$, $hd(M_2) \geq 4$ and $hd(M_1,M_2) \geq 3$,
so $\Pi$ is a distance-4 partition system.
By Theorem \ref{th:simple_pe}, $hd(ext(\Pi)) \geq 4$. 

\begin{table} 
\scriptsize
\centering
\setlength{\textfloatsep}{0.1cm}
\begin{tabular}{cc}
     \textsc{Initial Permutations in $\Pi$} & \textsc{Modified Permutations in $ext(\Pi)$} \\
      {$\!\begin{aligned}
        \arraycolsep=2.5pt
M_1=        
        \left[\begin{array}
          {cccc}
\color{blue}\mathbf{0} & 1 & 2 & 3 \\
\color{blue}\mathbf{1} & 0 & 3 & 2 \\
2 & 3 & \color{blue}\mathbf{0} & 1  \\
3 & 2 & \color{blue}\mathbf{1} & 0  \\
    \end{array}\right]\\
        
        \arraycolsep=2.5pt
M_2=        
        \left[\begin{array}
          {cccc}
0 & \color{blue}\mathbf{2} & 3 & 1 \\
1 & \color{blue}\mathbf{3} & 2 & 0 \\
2 & 0 & 1 & \color{blue}\mathbf{3} \\
3 & 1 & 0 & \color{blue}\mathbf{2} \\
    \end{array}\right]\\
    
        \arraycolsep=2.5pt
M_3=        
        \left[\begin{array}
          {cccc}
0 & 3 & 1 & 2 \\
1 & 2 & 0 & 3 \\
2 & 1 & 3 & 0 \\
3 & 0 & 2 & 1 \\
    \end{array}\right]\\
        
      \end{aligned}$}
      
        & {$
        \!
        \begin{aligned}
        \arraycolsep=2.5pt
ext(M_1)=        
        \left[\begin{array}
          {ccccc}
\color{red}\mathbf{4} & 
1 & 2 & 3 &
\color{blue}\mathbf{0}\\

\color{red}\mathbf{4} & 
0 & 3 & 2 &
\color{blue}\mathbf{1}\\

2 & 3 & 
\color{red}\mathbf{4} & 1 & 
\color{blue}\mathbf{0}  
\\
3 & 2 & 
\color{red}\mathbf{4} & 0 &  
\color{blue}\mathbf{1} \\
    \end{array}\right]\\

        \arraycolsep=2.5pt
ext(M_2)=        
        \left[\begin{array}
          {ccccc}
0 & 
\color{red}\mathbf{4} & 
3 & 1 &
\color{blue}\mathbf{2}\\

1 & 
\color{red}\mathbf{4} & 
2 & 0 &
\color{blue}\mathbf{3}\\

2 & 0 & 1 &
\color{red}\mathbf{4} &  
\color{blue}\mathbf{3}  \\

3 & 1 & 0 &
\color{red}\mathbf{4} &  
\color{blue}\mathbf{2}  \\

    \end{array}\right]\\
        
        \arraycolsep=2.5pt
ext(M_3)=        
        \left[\begin{array}
          {ccccc}
0 & 3 & 1 & 2 & \color{red}\mathbf{4} \\
1 & 2 & 0 & 3 & \color{red}\mathbf{4} \\
2 & 1 & 3 & 0 & \color{red}\mathbf{4} \\
3 & 0 & 2 & 1 & \color{red}\mathbf{4} \\
    \end{array}\right]\\
        
  \end{aligned}
  $
  }
\end{tabular}
\caption{An example of simple partition and extension on the distance-4 partition system $\Pi=(\M,\P,\Q)$, where $\M=(M_1,M_2,M_3),~ \P=\{ \{0,2\},~\{1,3 \} \}$ and $\Q=\{ \{0,1\},~\{2,3 \} \}$.
The column on the left shows the ordered list of PAs $\M$ consisting three PAs, $M_1$, $M_2$ and $M_3$  on $Z_4$ with $hd(M_i) \geq 4$, for $i \in \{1,2,3\}$, and $hd(\M) \geq 3$. The column on the right shows the new PAs, $ext(M_1)$, $ext(M_2)$ and $ext(M_3)$, obtained by simple partition and extension.  
By Theorem \ref{th:simple_pe}, $hd(ext(\Pi)) \geq 4$.}
\label{tb:toy_pe_ex}
\end{table}

\section{Sequential Partition and Extension}
\label{s:seq_pe}

Let $\mathfrak M=\{M_1, M_2, \dots M_t\}$, for some $t$, be a collection of PAs on $Z_n$ that satisfy Properties \ref{dist-d prop1} and \ref{dist-d prop2 } for a distance-$d$ partition system. 
The basic idea of sequential partition and extension is that we first create several disjoint PA's by simple partition and extension, each consisting of permutations on $n+1$ symbols with internal Hamming distance $d$. Then, we use partition and extension again on these PA's to get a larger PA on $n+2$ symbols and Hamming distance $d$.
Such an iterative application of partition and extension can produce interesting new results. 

Let $(\M_1,\M_2,\dots,\M_m)$ be an ordered set of subsets of $\mathfrak M$ such that
each $\M_i$ contains some number of PAs, such as $M_k,\dots,M_l$, from $\mathfrak M$, and for all $i,j, ~ (1\leq i<j \leq m)$,  ${\mathcal M}_i$ and ${\mathcal M}_j$ are pairwise disjoint. 
Let \{$\Pi_1,\Pi_2,\dots,\Pi_m$\}, be a collection of  distance-$d$ partition systems on $Z_n$, where for all $i, ~ (1\leq i \leq m), ~\Pi_i = (\mathcal{M}_i, \mathcal{P}_i,\mathcal{Q}_i)$, and  
${\mathcal M}_i \subseteq \mathfrak M$. 
We say that  \{$\Pi_1,\Pi_2,\dots,\Pi_m$\} is \textit{pairwise disjoint} if for all $i,j, ~ (1\leq i<j \leq m)$,  ${\mathcal M}_i$ and ${\mathcal M}_j$ are pairwise disjoint. 

For each iteration $i$, we employ a different distance-$d$ partition system, $\Pi_i=({\mathcal M}_i,{\mathcal P}_i,{\mathcal Q}_i)$, that uses a previously unused set of PAs, 
${\mathcal M}_i \subseteq \mathfrak M$, 
to create a new PA, ext($\Pi_i$), on $Z_{n+1}$, with Hamming distance $d$. 
Hence, by repeated simple partition and extension, we create a collection of new PAs, ${ext(\Pi_1), ext(\Pi_2), ... , ext(\Pi_m)}$, for some $m>1$. 
As long as the distance-$d$ partition systems ${\Pi_1,\Pi_2,\dots,\Pi_m}$ are pairwise disjoint, the sets $\{ext(\Pi_1), ext(\Pi_2), ... , ext(\Pi_m)\}$ are pairwise disjoint as well. 

In the following, we assume that the distance-$d$ partition systems under consideration are pairwise disjoint.
The partitions ${\mathcal P}_i$ and ${\mathcal Q}_i$ 
need not be distinct from partitions  
${\mathcal P}_j$ and ${\mathcal Q}_j$. 

Consider the case of applying simple partition and extension twice in succession using two distance-$d$ partition systems, $\Pi_1=({\mathcal M}_1,{\mathcal P}_1,{\mathcal Q}_1)$ and  $\Pi_2=({\mathcal M}_2,{\mathcal P}_2,{\mathcal Q}_2)$.
We present Theorem \ref{th:iterative_pe_d2} and Corollary \ref{cor:iterative_pe_size}, which give results on the Hamming distance and the size of the resulting PA. Corollary \ref{cor:iterative_pe_general} extends these results by induction. These results will be useful later for describing a new method for creating PAs which we call \textit{sequential partition and extension}.

\begin{theorem}
\label{th:iterative_pe_d2}
Let $\Pi_1 = ({\mathcal M}_1, \mathcal{P}_1,\mathcal{Q}_1)$ and $\Pi_2 = ({\mathcal M}_2, \mathcal{P}_2,\mathcal{Q}_2)$ be pairwise disjoint distance-$d$ partition systems for $Z_n$, with 
$hd({\mathcal M}_1, {\mathcal M}_2) \geq d-1$. Then $hd(ext(\Pi_1)) \geq d, hd(ext(\Pi_2)) \geq d$, and $hd(ext(\Pi_1), ext(\Pi_2)) \geq d-1$. 
\end{theorem}

\begin{proof}
By Theorem \ref{th:simple_pe}, $hd(ext(\Pi_1)) \geq d$, $hd(ext(\Pi_2)) \geq d$. We show that $hd(ext(\Pi_1), ext(\Pi_2)) \geq d-1$. Pick two arbitrary permutations $\sigma'\in ext(\Pi_1)$ and  $\tau' \in ext(\Pi_2)$, where for some $k$ and $j$, $\sigma' = ext_k(\sigma)$ for some $\sigma \in \Pi_1$, and $\tau' = ext_j(\tau)$ for some $\tau \in \Pi_2$. We consider two cases to determine the number of new agreements between $\sigma'$ and $\tau'$ created by the extension operation:
\begin{enumerate} [{Case} 1{:}]
\item $k = j$
\\The extension operation creates a new agreement in position $k=j$ because $\sigma'(k) = \tau'(k) = n$. Note that since $\sigma'(n) = \sigma(k)$ and $\tau'(n)=\tau (k)$, the relationship between $\sigma'(n)$ and $\tau'(n)$ is the same as the relationship between $\sigma(k)$ and $\tau(k)$. Hence, there is at most one new agreement between $\sigma'$ and $\tau'$. 

\item $k \neq j$
\\In this case, $\sigma'(k) = n$ and $\tau'(j)=n$, so the new symbol $n$ is in different positions in $\sigma'$ and $\tau'$. That is, inserting the symbol $n$ does not, in itself, increase the number of agreements. Now consider the symbols $\sigma(k)$ and $\tau(j)$.
If  $\sigma(k) = \tau(j)$, then $\sigma'(n) = \tau'(n)$. 
In this situation, extension creates a new agreement in position $n$. 
On the other hand, if $\sigma(k) \neq \tau(j)$, then $\sigma'(n) \neq \tau'(n)$, so no new agreement is created by extension. 
In either situation, extension creates at most one new agreement between $\sigma'$ and $\tau'$.
\end{enumerate}

By assumption, $hd({\mathcal M}_1, {\mathcal M}_2) \geq d-1$, hence $hd(\sigma, \tau) \geq d-1$ as well. 
That is the number of disagreements between $\sigma$ and $\tau$ is at least $d-1$, or equivalently, the number of agreements between $\sigma$ and $\tau$ is at most $n-(d-1)$. 
So, the number of agreements between $\sigma'$ and $\tau'$ is at most $1+ n-(d-1)$. 
Since $\sigma'=ext_k(\sigma)$ and $\tau' = ext_m(\tau)$, both $\sigma'$ and $\tau'$ are permutations on $n+1$ (not $n$) symbols. 
Hence, $hd(\sigma',\tau')\geq (n+1)-(1+n-(d-1))\geq d-1$, so $hd(ext(\Pi_1), ext(\Pi_2)) \geq d-1$.\\
\end{proof}

\begin{cor} \label{cor:iterative_pe_size}
Let $\Pi_1 = ({\mathcal M}_1, \mathcal{P}_1,\mathcal{Q}_1)$ and $\Pi_2 = ({\mathcal M}_2, \mathcal{P}_2,\mathcal{Q}_2)$
be pairwise disjoint distance-$d$ partition systems for $Z_n$,
with $hd({\mathcal M}_1,{\mathcal M}_2)\geq d-1$. 
Let ${\mathcal A} = ext(\Pi_1) \cup ext(\Pi_2)$. Then ${\mathcal A}$ is a PA on $Z_{n+1}$ such that $|{\mathcal A}|=|ext(\Pi_1)| + |ext(\Pi_2)|$ and $hd({\mathcal A})\ge d-1$.

\end{cor}

\begin{proof}
Since both $ext(\Pi_1)$ and $ext(\Pi_2)$ are created by simple partition and extension of PAs on $Z_n$, ${\mathcal A}$ is a PA on $Z_{n+1}$. 
Given that ${\mathcal M}_1$ is disjoint from ${\mathcal M}_2$, Equation \ref{eq:sizeof_extPi} tells us that $|{\mathcal A}|=|ext(\Pi_1)| + |ext(\Pi_2)| $.
Lastly, by Theorem \ref{th:iterative_pe_d2}, $hd({\mathcal A})\ge d-1$.
\end{proof}

\bigskip
Simple partition and extension can be used in a similar way on several more distance-$d$ partition systems on $Z_n$ to create large PAs on $Z_{n+1}$. This is formalized by Corollary \ref{cor:iterative_pe_general}. 

\begin{cor} \label{cor:iterative_pe_general}
Let $\Pi_1 = (\mathcal{M}_1, \mathcal{P}_1,\mathcal{Q}_1)$, $\Pi_2 = (\mathcal{M}_2, \mathcal{P}_2,\mathcal{Q}_2)$, $\dots$ , $\Pi_m = (\mathcal{M}_m, \mathcal{P}_m,\mathcal{Q}_m)$ be a collection of pairwise disjoint distance-$d$ partition systems, for some $m>1$, where $hd(\mathcal{M}_i, \mathcal{M}_j) \geq d-1$, for all $i,j~ (1 \le i<j \le m)$. Let ${\mathcal A} = ext(\Pi_1) \cup ext(\Pi_2) \cup \dots \cup ext(\Pi_m)$. 
Then 
\begin{enumerate} [\indent(1)]
\item $\forall i,j~ (1 \le i<j \le m),~ hd(ext(\Pi_i), ext(\Pi_j)) \geq d-1$,
\item ${\mathcal A}$ is a PA on $Z_{n+1}$,  
\item $|{\mathcal A}|=\sum_{i=1}^m |ext(\Pi_i)|$, and 
\item $hd({\mathcal A}) \geq d-1$.
\end{enumerate}
\end{cor}
\begin{proof}
The results follow from Theorem \ref{th:iterative_pe_d2} and Corollary \ref{cor:iterative_pe_size} by induction on $m$.
\end{proof}

\bigskip
A new technique, which we call 
\textit{sequential partition and extension}, can be used to improve bounds for $M(n+2,d)$. It has two steps. 
First, simple partition and extension is used to create the extended PAs ${ext(\Pi_1), ext(\Pi_2), ... , ext(\Pi_m)}$, for some $m>1$.  Let ${\mathbb M}=\{ {\mathbb M}_1,{\mathbb M}_2, \dots, {\mathbb M}_m \}$, where for all $i$, 
${\mathbb M}_i = ext(\Pi_i)$. Note that ${\mathbb M}$ is a collection of PAs on $Z_{n+1}$.
Let $\mathbb{P}$ and $\mathbb{Q}$ be partitions of $Z_{n+1}$ such that $\Psi = (\mathbb {M},\mathbb {P}, \mathbb {Q})$ is a distance-$d$ partition system on $Z_{n+1}$. 
Next, simple partition and extension is again used to create a new PA, $ext(\Psi)$, on $Z_{n+2}$. 

We show that $ext(\Psi)$ is a PA on $n+2$ symbols with Hamming distance $d$.

\begin{theorem}
\label{th:sequential_pe}
Sequential partition and extension on a collection 
\{$\Pi_1,\Pi_2,\dots,\Pi_m$\}, of pairwise disjoint  distance-$d$ partition systems on $Z_n$, 
results in a new PA on $Z_{n+2}$ with Hamming distance $d$.
\end{theorem}

\begin{proof}
Let ${ext(\Pi_1), ext(\Pi_2), ... , ext(\Pi_m)}$ be the PAs on $Z_{n+1}$ created the first phase of sequential partition and extension. 
By Theorem \ref{th:simple_pe}, $hd(ext(\Pi_i)) \ge d$.
By Corollary \ref{cor:iterative_pe_general}, $ \forall i,j~(1 \le i<j \le m),~ 
{hd(ext(\Pi_i), ext(\Pi_j)) \ge d-1}$.

Let $\mathbb {M}=(ext(\Pi_1),~ext(\Pi_2),~\dots,~ext(\Pi_m))$, and 
let $\mathbb {P}$ and $\mathbb {Q}$ be suitable partitions of $Z_{n+1}$, 
such that $\Psi = (\mathbb {M},\mathbb {P}, \mathbb {Q})$ forms a distance-$d$ partition system on $Z_{n+1}$. 
Let $ext(\Psi)$ be the PA created by simple partition and extension on $\Psi = (\mathbb {M},\mathbb {P}, \mathbb {Q})$.
Since, $\Psi$ is a distance-$d$ partition system on $Z_{n+1}$, $ext(\Psi)$ is a PA on $Z_{n+2}$. 
By Theorem \ref{th:simple_pe}, $hd(ext(\Psi)) \geq d$. 
\end{proof}

We now illustrate sequential partition and extension by means of an example.  

\begin{changemargin}{1cm}{1cm}
\begin{example}
\label{example_sequential_pe}Consider the group $AGL(1,37)$ on $37$ symbols with Hamming distance $36$, containing $1332$ permutations. 
This gives $M(37,36) \geq 1332$. Using sequential partition and extension we show that $M(39,37) \geq 1301$. 

$AGL(1,37)$ can be decomposed into $36$ Latin squares, 
where one of the Latin squares is a cyclic subgroup of $AGL(1,37)$ consisting of the identity permutation and all cyclic shifts. 
This is the set of permutations $C_1 =\{ x+b ~| ~ b \in Z_{37} \}.$  
The other $35$ Latin squares can be defined as the left cosets of $C_1$, namely, $C_i = \{ ix+b ~| ~ b \in Z_{37} \}$, for each $i ~ (2 \leq i \leq  36)$.

\begin{sloppy} 
First, we give six distance-$37$ partition systems for $AGL(1,37)$, namely, 
${\Pi_1 = ( {\mathcal M}_1, {\mathcal P}_1, {\mathcal Q}_1)}$, 
${\Pi_2 = ({\mathcal M}_2, {\mathcal P}_2, {\mathcal Q}_2)}$, 
${\Pi_3 = ({\mathcal M}_3, {\mathcal P}_3, {\mathcal Q}_3)}$, 
${\Pi_4 = ({\mathcal M}_4, {\mathcal P}_4, {\mathcal Q}_4)}$, 
${\Pi_5 = ({\mathcal M}_5, {\mathcal P}_5, {\mathcal Q}_5)}$,
${\Pi_6 = ({\mathcal M}_6, {\mathcal P}_6, {\mathcal Q}_6)}$, where 
${{\mathcal M}_1 = \{C_1, C_2, \dots , C_7\}}$, 
${{\mathcal M}_2 = \{C_8, C_9, \dots , C_{14}\}}$, 
${{\mathcal M}_3 = \{C_{15}, C_{16}, \dots , C_{21}\}}$, 
${{\mathcal M}_4 = \{C_{22}, C_{23}, \dots , C_{28}\}}$, 
${{\mathcal M}_5 = \{C_{29}, C_{30}, \dots , C_{35}\}}$, 
${{\mathcal M}_6 = \{C_{36}\}}$ 
with the partitions 
$\mathcal{P}_i$, $\mathcal{Q}_i$  $(1 \leq i \leq 6 )$ described in Table \ref{tb:seq_pe1}. 
Note that in each $\Pi_i$, the last coset is covered by adding the new symbol '37' in the $37^{th}$ position. 

Simple partition and extension 
yields six PAs on $Z_{38}$, where 
for all $i,~ (1 \le i \le 6),~ hd(ext(\Pi_i)) \geq 37 $, 
and  for all $i,j~ (1 \le i<j \le 6),~ hd(ext(\Pi_i), ext(\Pi_j)) \geq 36$. Moreover,
$|ext(\Pi_1)| = 253$, 
$|ext(\Pi_2)| = 253$, 
$|ext(\Pi_3)| = 253$, 
$|ext(\Pi_4)| = 253$, 
$|ext(\Pi_5)| = 252$, and 
$|ext(\Pi_6)|$ = 37.

\end{sloppy}

\begin{sloppy}
Finally, 
we form a distance-$37$ partition system 
${\Psi=(\mathbb{M}, \mathbb{P}, \mathbb{Q})}$, 
where $\mathbb{M}=(ext(\Pi_1), ext(\Pi_2), \dots , ext(\Pi_6))$ 
with suitable partitions $\mathbb{P}$ and $\mathbb{Q}$ as shown in Table \ref{tb:seq_pe2}. 
The result is a PA, $ext(\Psi)$, on $39$ symbols with Hamming distance $37$, which has $1301$ permutations. 
The previous lower bound for $M(39,37)$, given by the five known MOLS on $39$ symbols, was $195$.

\end{sloppy}

\end{example}
\end{changemargin}

\bigskip
Sequential partition and extension also results in the lower bounds 
$M(34,32) \geq 945$ and $M(66,64) \geq 4029$. 
Table \ref{tb:n_2} shows additional improved lower bounds on $M(n,n-2)$ obtained by sequential partition and extension.

\begin{table} [htb]
\centering
\begin{tabular}{|c|c|c|c|c|}
\hline
$\Pi_i$ & Set of Cosets, $\mathcal{M}_i$ & $\mathcal{P}_i$ & $\mathcal{Q}_i$ & $|ext(\Pi_i)|$\\
\hline
\hline
\multirow{6}{*}{$\Pi_1$} & $\{x+b ~|~ b \in Z_{37}\}$ & $\{4, 11, 18, 25, 31, 34\}$ & $\{0, 1, 2, 3, 4, 5, 6\}$ & \multirow{6}{*}{$253$} \\
& $\{2x+b ~|~ b \in Z_{37}\}$ & $\{5, 8, 10, 13, 16, 19, 21\}$  & $\{7, 8, 9, 10, 11, 12\}$ & \\
& $\{3x+b ~|~ b \in Z_{37}\}$ & $\{14, 20, 22, 24, 28, 30\}$  & $\{13, 14, 15, 16, 17, 18\}$ & \\
& $\{4x+b ~|~ b \in Z_{37}\}$ & $\{9, 12, 15, 26, 29, 32\}$  & $\{19, 20, 21, 22, 23, 24\}$ & \\
& $\{5x+b ~|~ b \in Z_{37}\}$ & $\{6, 7, 17, 23, 27, 33\}$  & $\{25, 26, 27, 28, 29, 30\}$ & \\
& $\{6x+b ~|~ b \in Z_{37}\}$ & $\{0, 1, 2, 3, 35, 36\}$  & $\{31, 32, 33, 34, 35, 36\}$ & \\
& $\{7x+b ~|~ b \in Z_{37}\}$ & $\{37\}$  & $\{37\}$ & \\
\hline
\multirow{6}{*}{$\Pi_2$} & $\{8x+b ~|~ b \in Z_{37}\}$ & $\{1, 12, 23, 25, 36\}$ & $\{0, 1, 2, 3, 4, 5, 6\}$ & \multirow{6}{*}{$253$} \\
& $\{9x+b ~|~ b \in Z_{37}\}$ & $\{0, 11, 13, 22, 24, 35\}$ & $\{7, 8, 9, 10, 11, 12\}$ & \\
& $\{10x+b ~|~ b \in Z_{37}\}$ & $\{8, 9, 10, 17, 18, 26, 27\}$ & $\{13, 14, 15, 16, 17, 18\}$ & \\
& $\{11x+b ~|~ b \in Z_{37}\}$ & $\{4, 5, 6, 7, 19, 20, 28\}$ & $\{19, 20, 21, 22, 23, 24\}$ & \\
& $\{12x+b ~|~ b \in Z_{37}\}$ & $\{14, 15, 16, 32, 33, 34\}$ & $\{25, 26, 27, 28, 29, 30\}$ & \\
& $\{13x+b ~|~ b \in Z_{37}\}$ & $\{2, 3, 21, 29, 30, 31\}$ & $\{31, 32, 33, 34, 35, 36\}$ & \\
& $\{14x+b ~|~ b \in Z_{37}\}$ & $\{37\}$  & $\{37\}$ & \\
\hline
\multirow{6}{*}{$\Pi_3$} & $\{15x+b ~|~ b \in Z_{37}\}$ & $\{2, 3, 4, 6, 15, 27\}$ & $\{0, 1, 2, 3, 4, 5, 6\}$ & \multirow{6}{*}{$253$} \\
& $\{16x+b ~|~ b \in Z_{37}\}$ & $\{12, 13, 14, 16, 17, 18, 22\}$ & $\{7, 8, 9, 10, 11, 12\}$ & \\
& $\{17x+b ~|~ b \in Z_{37}\}$ & $\{0, 21, 25, 28, 29, 33\}$ & $\{13, 14, 15, 16, 17, 18\}$ & \\
& $\{18x+b ~|~ b \in Z_{37}\}$ & $\{7, 8, 19, 20, 31, 32\}$ & $\{19, 20, 21, 22, 23, 24\}$ & \\
& $\{19x+b ~|~ b \in Z_{37}\}$ & $\{10, 11, 23, 24, 35, 36\}$ & $\{25, 26, 27, 28, 29, 30\}$ & \\
& $\{20x+b ~|~ b \in Z_{37}\}$ & $\{1, 5, 9, 26, 30, 34\}$ & $\{31, 32, 33, 34, 35, 36\}$ & \\
& $\{21x+b ~|~ b \in Z_{37}\}$ & $\{37\}$  & $\{37\}$ & \\
\hline
\multirow{6}{*}{$\Pi_4$} & $\{22x+b ~|~ b \in Z_{37}\}$ & $\{2, 3, 5, 9, 21, 33\}$ & $\{0, 1, 2, 3, 4, 5, 6\}$ & \multirow{6}{*}{$253$} \\
& $\{23x+b ~|~ b \in Z_{37}\}$ & $\{4, 8, 11, 22, 23, 34\}$ & $\{7, 8, 9, 10, 11, 12\}$ & \\
& $\{24x+b ~|~ b \in Z_{37}\}$ & $\{7, 16, 17, 25, 26, 35\}$ & $\{13, 14, 15, 16, 17, 18\}$ & \\
& $\{25x+b ~|~ b \in Z_{37}\}$ & $\{12, 13, 14, 30, 31, 32\}$ & $\{19, 20, 21, 22, 23, 24\}$ & \\
& $\{26x+b ~|~ b \in Z_{37}\}$ & $\{1, 6, 10, 15, 24, 29\}$ & $\{25, 26, 27, 28, 29, 30\}$ & \\
& $\{27x+b ~|~ b \in Z_{37}\}$ & $\{0, 18, 19, 20, 27, 28, 36\}$ & $\{31, 32, 33, 34, 35, 36\}$ & \\
& $\{28x+b ~|~ b \in Z_{37}\}$ & $\{37\}$  & $\{37\}$ & \\
\hline
\multirow{6}{*}{$\Pi_5$} & $\{29x+b ~|~ b \in Z_{37}\}$ & $\{2, 5, 13, 18, 26, 29\}$ & $\{0, 1, 2, 3, 4, 5, 6\}$ & \multirow{6}{*}{$252$} \\
& $\{30x+b ~|~ b \in Z_{37}\}$ & $\{12, 19, 21, 27, 34, 36\}$ & $\{7, 8, 9, 10, 11, 12\}$ & \\
& $\{31x+b ~|~ b \in Z_{37}\}$ & $\{6, 7, 8, 9, 10, 11\}$ & $\{13, 14, 15, 16, 17, 18\}$ & \\
& $\{32x+b ~|~ b \in Z_{37}\}$ & $\{4, 14, 15, 25, 31, 35\}$ & $\{19, 20, 21, 22, 23, 24\}$ & \\
& $\{33x+b ~|~ b \in Z_{37}\}$ & $\{0, 3, 16, 17, 20, 23, 33\}$ & $\{25, 26, 27, 28, 29, 30\}$ & \\
& $\{34x+b ~|~ b \in Z_{37}\}$ & $\{1, 22, 24, 28, 30, 32\}$ & $\{31, 32, 33, 34, 35, 36\}$ & \\
& $\{35x+b ~|~ b \in Z_{37}\}$ & $\{37\}$  & $\{37\}$ &\\
\hline
$\Pi_6$ & $\{36x+b ~|~ b \in Z_{37}\}$ & $\{37\}$ & $\{37\}$  & 37 \\ \hline
\end{tabular}
\caption{Step 1 of Sequential Partition and Extension on $AGL(1,37)$, which gives $M(38,36) \geq 1301$.}
\label{tb:seq_pe1}
\end{table}

In fact, sequential partition and extension can be applied an arbitrary number of times, provided that suitable distance-$d$ partitions systems can be found at each stage. That is, sequential partition and extension on a sequence of $r$ distance-$d$ partitions systems could result in new lower bounds for $M(n+r,d)$, for arbitrary $r$.

\clearpage

\begin{table} [h]
\centering
\begin{tabular}{|c|c|c|c|}
\hline
${\mathbb M}$ & ${\mathbb P}_i\in {\mathbb P}$ & ${\mathbb Q}_i \in {\mathbb Q}$ & $|ext({\mathbb M}_i)|$\\
\hline
\hline
${\mathbb M}_1$=ext($\Pi_1$)
& $\{4, 11, 18, 25, 31, 34\}$ & $\{0, 1, 2, 3, 4, 5, 6\}$ & $253$ \\
${\mathbb M}_2$=ext($\Pi_2$) & $\{5, 8, 10, 13, 16, 19, 21\}$  & $\{7, 8, 9, 10, 11, 12\}$ & $253$ \\
${\mathbb M}_3$=ext($\Pi_3$) & $\{14, 20, 22, 24, 28, 30\}$  & $\{13, 14, 15, 16, 17, 18\}$ & $253$ \\
${\mathbb M}_4$=ext($\Pi_4$) & $\{9, 12, 15, 26, 29, 32\}$  & $\{19, 20, 21, 22, 23, 24\}$ & $253$ \\
${\mathbb M}_5$=ext($\Pi_5$) & $\{38\}$ & $\{38\}$ & $252$ \\
${\mathbb M}_6$=ext($\Pi_6$) & $\{0,1,2,3,6,7,17,23,27,$  & $\{25,26,27,28,29,30,31$ &   \\
 & $33,35,36,37\}$  & $32,33,34,35,36,37\}$ & $37$ \\
\hline
\hline
Total &  &  & $1301$ \\
\hline
\end{tabular}
\caption 
{Step 2 of Sequential Partition and Extension on $AGL(1,37)$ for $M(39,37) \ge 1301$.} 
\label{tb:seq_pe2}
\end{table}

\begin{table} [h!]
\centering
\vspace*{2mm}
\begin{tabular}{|r r r | r r r | r r r|}
\hline
$n$ & \bf PREV & \bf NEW & $n$ & \bf PREV & \bf NEW & 
$n$ & \bf PREV & \bf NEW \\
\hline\hline

34	 & 	 192 	 & 	 945 	 & 	 159 	 & 	 2,051 	 & 	 16,666 	 & 	 291 	 & 	 5,202 	 & 	 80,385 	 \\
39	 & 	 255 	 & 	 1,301 	 & 	 165 	 & 	 2,185 	 & 	 17,632 	 & 	 295 	 & 	 5,088 	 & 	 54,572 	 \\
45	 & 	 270 	 & 	 1,726 	 & 	 171 	 & 	 2,354 	 & 	 27,330 	 & 	 309 	 & 	 5,539 	 & 	 60,715 	 \\
51	 & 	 392 	 & 	 2,308 	 & 	 175 	 & 	 2,354 	 & 	 19,792 	 & 	 315 	 & 	 5,634 	 & 	 60,952 	 \\
55	 & 	 423 	 & 	 2,461 	 & 	 183 	 & 	 2,533 	 & 	 21,994 	 & 	 319 	 & 	 5,793 	 & 	 67,379 	 \\
63	 & 	 1,514 	 & 	 3,306 	 & 	 195 	 & 	 2,758 	 & 	 25,022 	 & 	 333 	 & 	 6,091 	 & 	 70,696 	 \\
66	 & 	 576 	 & 	 4,029 	 & 	 201 	 & 	 2,867 	 & 	 25,427 	 & 	 339 	 & 	 6,280 	 & 	 69,485 	 \\
69	 & 	 594 	 & 	 3,965 	 & 	 213 	 & 	 3,170 	 & 	 30,288 	 & 	 345 	 & 	 5,205 	 & 	 89,272 	 \\
75	 & 	 667 	 & 	 4,747 	 & 	 225 	 & 	 3,421 	 & 	 32,728 	 & 	 351 	 & 	 6,642 	 & 	 76,195       \\
85	 & 	 812 	 & 	 6,116 	 & 	 231 	 & 	 3,548 	 & 	 33,779 	 & 	 355 	 & 	 6,746 	 & 	 77,215       \\
91	 & 	 902 	 & 	 6,709 	 & 	 235 	 & 	 3,625 	 & 	 35,001 	 & 	 363 	 & 	 7,220 	 & 	 125,709 	 \\
99	 & 	 1,017 	 & 	 8,206 	 & 	 245 	 & 	 3,475 	 & 	 43,717 	 & 	 369 	 & 	 7,108 	 & 	 83,418 	 \\
105	 & 	 1,119 	 & 	 9,239 	 & 	 253 	 & 	 4,075 	 & 	 40,094 	 & 	 375 	 & 	 7,298 	 & 	 87,434 	 \\
111	 & 	 1,187 	 & 	 9,990 	 & 	 259 	 & 	 4,222 	 & 	 43,268 	 & 	 385 	 & 	 7,428 	 & 	 90,213 	 \\
115	 & 	 1,277 	 & 	 11,142  & 	 265 	 & 	 4,342 	 & 	 44,733 	 & 	 391 	 & 	 7,690 	 & 	 90,991 	 \\
123	 & 	 1,452 	 & 	 13,996  & 	 273 	 & 	 4,548 	 & 	 46,268 	 & 	 411 	 & 	 8,240 	 & 	 104,098 	 \\
133	 & 	 1,554 	 & 	 11,604  & 	 279 	 & 	 4,701 	 & 	 49,243 	 & 	 514 	 & 	 11,264  & 	 197,859 	 \\
141	 & 	 1,723 	 & 	 13,522  & 	 285 	 & 	 4,868 	 & 	 51,571 	 & 	 531 	 & 	 12,696  & 	 271,043 	 \\
153	 & 	 1,923 	 & 	 16,118  & 		     &   		 & 		         & 		     & 	    	 & 	         	 \\
\hline
\end{tabular}
\caption{$M(n,n-2)$ lower bounds. \textbf{PREV} denotes the previous bound and \textbf{NEW} denotes the new bound obtained using Sequential Partition and Extension.}
\label{tb:n_2}
\end{table}

\section{Parallel Partition and Extension} \label{s:par_pe}

In Section \ref{s:seq_pe}, we described a new technique, based on simple partition and extension, called sequential partition and extension. We now present another new technique, called \textit{parallel partition and extension} which introduces multiple new symbols simultaneously. 
As previously described, 
simple partition and extension extends a permutation array by replacing \textit{one} existing symbol in a carefully selected position in each permutation with the symbol $n$, and 
appending the displaced symbol to the end of the permutation. 
Sequential partition and extension allows additional symbols to be introduced one at a time by applying simple partition and extension sequentially.
In contrast, \textit{parallel partition and extension} on a PA $A$ on $Z_n$ creates a PA $A'$ on $Z_{n+r}$ by introducing, to each permutation in $A$, $r$ new symbols \textit{simultaneously}. 
Table \ref{tb:new_parallel_pe} shows new bounds obtained using Theorems \ref{thm:rudi} and \ref{thm:parallel_pe_2} for parallel partition and extension. These theorems are proved in Sections \ref{rudi} and \ref{general_parallel} below.

\subsection{Rudimentary Parallel Partition and Extension} \label{rudi}

In its rudimentary form, parallel partition and extension operates on $2r$ \textit{blocks} (\ie{sets}) of permutations, for some integer $r$. 
Specifically, suppose a PA $A$, on $Z_n$, is partitioned into $k=2r$ blocks of permutations $B_0, B_1, \dots, B_{k-1}$, 
where, for all $i,~(0 \leq i < k)$, $hd(B_i)\geq d$, for some $d$, 
and for all $i,j~(0 \leq i \neq j < k),~ hd(B_i, B_j)\geq d-r$.  
In particular, $hd(A)\geq d-r$. 
We create a new PA $A'$ on $Z_{n+r}$, such that $hd(A')\geq d$, by inserting  a sequence of new symbols from the set  
$\{n, n+1, \dots, n+r-1\}$ into the permutations in each block.  Each block uses a different sequence.

Define SHIFT$(\gamma,0)$ to be the sequence $(n,n+1, n+2, \dots, n+r-1)$, and for each integer $t$,
denote by SHIFT$(\gamma,t)$ the left cyclic shift of the sequence by $t$ (mod $r$) positions. 
For example, SHIFT$(\gamma,1)$ is the sequence $(n+1, n+2, \dots, n+r-1,n)$, and SHIFT$(\gamma,2)$ is the sequence $(n+2, \dots , n+r-1, n, n+1)$, and so on. 

The creation of the new PA $A'$ takes place in two steps. The first step modifies the blocks $B_0, B_1, \dots, B_{r-1}$.
For all $l, ~ (0\leq l<r)$, a new block $B'_l$ of permutations on $Z_{n+r}$ is created from the block $B_l$ as follows:  the first $r$ symbols in each permutation of $B_l$, are replaced by SHIFT$(\gamma,l)$, and the $r$ replaced symbols are put in their original order at the end of the permutation in positions $n, n+1, \dots , n+r-1$. 

In the second step, a new block of permutations $B'_m$  is created from each block $B_m$, for all $m, ~ (r\leq m <2r)$, by appending the sequence, SHIFT$(\gamma,m)$ to each permutation in positions $n, n+1, \dots , n+r-1$.  
The blocks $B'_l,~ (0\leq l<r)$ together with the blocks $B'_m,~ (r\leq m <2r)$ comprise the new PA $A'$ on $Z_{n+r}$.

It is known that the Hamming distance between two permutations does not change when the order of the symbols in both permutations is altered in a fixed manner.  
Consequently, the Hamming distance between permutations in the same block, or between permutations in different blocks is not altered by the movement of the first $r$ symbols in each permutation to positions $n, n+1, \dots , n+r-1$. 
Since the ordering of the new symbols $n, n+1, \dots , n+r-1$ in any block is a cyclic shift of sequence of new symbols in any other block, rudimentary parallel partition and extension does not create any new agreements between permutations in different blocks. 
For the original permutation array $A,~ hd(A)\geq d-r$. 
For the new permutation array $A'$, the permutations in each block have been extended by $r$ symbols in a way that ensures that the inter-block Hamming distance is at least $d$.
That is, for all $i,j~(0 \leq i \neq j < k),~ hd(B'_i, B'_j)\geq d$, and the length of the permutations has increased by $r$. 
Within each new block, the $r$ new symbols are put in a fixed order into fixed positions, creating $r$ new agreements in addition to the $(n-d)$ agreements that existed in the unaltered blocks.
For the new blocks $B'_l$ for all $l ~ (0\leq l<r)$, the displaced symbols are moved to the end of each permutation. 
For the new blocks $B'_m$, for all $m ~ (r\leq m <2r)$, no symbols are displaced because the $r$ new symbols are appended at the end of the permutations.
Thus the intra-block Hamming distance for the new permutations is $(n+r-(r+(n-d)))=d$.
That is, for all $i,~(0 \leq i < k)$, $hd(B'_i)\geq d$. 
Hence, $hd(A')\geq d$.
The size of the PA $A'$ is given by Theorem \ref{thm:rudi}. The proof is described in \cite{nguyen2013transitivity}.

\begin{thm} [\cite{nguyen2013transitivity}] 
\label{thm:rudi}
Let $A$ be a PA on $Z_n$ comprising $2r$ blocks for some $r$. Denote the blocks by $B_0, B_1, \dots, B_{2r-1}$, so that $A=\cup_{i=0}^{2r-1} B_i$.
If each block $B_i$ has Hamming distance at least $d$ and 
the Hamming distance of the entire set $A$ is at least $d-r$, then rudimentary parallel partition and extension on $A$ results in a new PA $A'$ on $Z_{n+r}$ that exhibits $M(n+r,d)\ge \sum_{i=0}^{2r-1} |B_i|$.
\end{thm}

Table \ref{tb:par_pe_ex} illustrates rudimentary parallel partition and extension for $n=9, d=9$ and $r=3$ using a PA $A$ on $Z_9$. 
We provide $k=2r=6$ blocks such that for each block $B_i,~ (0 \leq i \leq 5),~ hd(B_i)\geq d=9$ and  for all $i,j ~(0\leq i \neq j \leq 5),~ hd(B_i, B_j)\geq d-r=6$. 
These blocks comprise the PA $A$ and are shown in the  column on the left of Table \ref{tb:par_pe_ex}. 
The symbols to be relocated by rudimentary parallel partition and extension are shown in blue.
Note that $hd(A) \geq 6$.
Rudimentary parallel partition and extension on $A$ results in the PA $A'$ on $Z_{12}$ with $hd(A') \geq 6$. The permutations comprising $A'$ are shown in the column on the right of Table \ref{tb:par_pe_ex}, with the displaced symbols shown in blue and the new symbols shown in red. 

More results based on Theorem \ref{thm:rudi} are shown in Table \ref{tb:new_parallel_pe}. 
For example, for $n=42, d=39, r=4$, take $PGL(2,41)$, which contains $40\cdot 41\cdot 42=68880$ permutations on 42 symbols, with hamming distance at least 39. 
We found $2r=8$ cosets of $PGL(2,41)$  with $d=35$. 
Then by Theorem \ref{thm:rudi}, $M(46,39)\ge 8\cdot 68880=551040$ using 8 cosets.

\begin{table} 
\scriptsize
\centering
\setlength{\textfloatsep}{0.1cm}
\begin{tabular}{cc}
     \textsc{Initial Permutations in the PA $A$} & \textsc{Modified Permutations in the PA $A'$} \\
      {$\!\begin{aligned}
        \arraycolsep=2.5pt
        \left[\begin{array}
          {ccccccccc}
\color{blue}\mathbf{0} & \color{blue}\mathbf{1} & 
	\color{blue}\mathbf{2} & 3 & 4 & 5 & 6 & 7 & 8 \\
\color{blue}\mathbf{1} & \color{blue}\mathbf{5} & 
	\color{blue}\mathbf{8} & 4 & 6 & 0 & 3 & 2 & 7 \\
\color{blue}\mathbf{2} & \color{blue}\mathbf{8} & 
	\color{blue}\mathbf{6} & 1 & 5 & 7 & 0 & 4 & 3 \\
\color{blue}\mathbf{3} & \color{blue}\mathbf{4} & 	
	\color{blue}\mathbf{1} & 7 & 2 & 6 & 8 & 0 & 5 \\
\color{blue}\mathbf{4} & \color{blue}\mathbf{6} & 
	\color{blue}\mathbf{5} & 2 & 8 & 3 & 7 & 1 & 0 \\
\color{blue}\mathbf{5} & \color{blue}\mathbf{0} & 	
	\color{blue}\mathbf{7} & 6 & 3 & 1 & 4 & 8 & 2 \\
\color{blue}\mathbf{6} & \color{blue}\mathbf{3} & 
	\color{blue}\mathbf{0} & 8 & 7 & 4 & 2 & 5 & 1 \\
\color{blue}\mathbf{7} & \color{blue}\mathbf{2} & 
	\color{blue}\mathbf{4} & 0 & 1 & 8 & 5 & 3 & 6 \\
\color{blue}\mathbf{8} & \color{blue}\mathbf{7} & 
	\color{blue}\mathbf{3} & 5 & 0 & 2 & 1 & 6 & 4
    
    \end{array}\right]\\
        
        \arraycolsep=2.5pt
        \left[\begin{array}
          {ccccccccc}
          
\color{blue}\mathbf{1} & \color{blue}\mathbf{3} & 		
	\color{blue}\mathbf{6} & 7 & 5 & 8 & 2 & 4 & 0 \\
\color{blue}\mathbf{5} & \color{blue}\mathbf{4} & 
	\color{blue}\mathbf{3} & 2 & 0 & 7 & 8 & 6 & 1 \\
\color{blue}\mathbf{8} & \color{blue}\mathbf{1} & 
	\color{blue}\mathbf{0} & 4 & 7 & 3 & 6 & 5 & 2 \\
\color{blue}\mathbf{4} & \color{blue}\mathbf{7} & 
	\color{blue}\mathbf{8} & 0 & 6 & 5 & 1 & 2 & 3 \\
\color{blue}\mathbf{6} & \color{blue}\mathbf{2} & 
	\color{blue}\mathbf{7} & 1 & 3 & 0 & 5 & 8 & 4 \\
\color{blue}\mathbf{0} & \color{blue}\mathbf{6} & 
	\color{blue}\mathbf{4} & 8 & 1 & 2 & 7 & 3 & 5 \\
\color{blue}\mathbf{3} & \color{blue}\mathbf{8} & 
	\color{blue}\mathbf{2} & 5 & 4 & 1 & 0 & 7 & 6 \\
\color{blue}\mathbf{2} & \color{blue}\mathbf{0} & 
	\color{blue}\mathbf{5} & 3 & 8 & 6 & 4 & 1 & 7 \\
\color{blue}\mathbf{7} & \color{blue}\mathbf{5} & 	
	\color{blue}\mathbf{1} & 6 & 2 & 4 & 3 & 0 & 8         

        \end{array}\right]\\
        
        \arraycolsep=2.5pt
        \left[\begin{array}
          {ccccccccc}
          
\color{blue}\mathbf{3} & \color{blue}\mathbf{5} & 
	\color{blue}\mathbf{7} & 2 & 6 & 0 & 8 & 4 & 1 \\
\color{blue}\mathbf{4} & \color{blue}\mathbf{0} & 
	\color{blue}\mathbf{2} & 8 & 3 & 1 & 7 & 6 & 5 \\
\color{blue}\mathbf{1} & \color{blue}\mathbf{7} & 
	\color{blue}\mathbf{4} & 6 & 0 & 2 & 3 & 5 & 8 \\
\color{blue}\mathbf{7} & \color{blue}\mathbf{6} & 
	\color{blue}\mathbf{0} & 1 & 8 & 3 & 5 & 2 & 4 \\
\color{blue}\mathbf{2} & \color{blue}\mathbf{3} & 
	\color{blue}\mathbf{1} & 5 & 7 & 4 & 0 & 8 & 6 \\
\color{blue}\mathbf{6} & \color{blue}\mathbf{1} & 
	\color{blue}\mathbf{8} & 7 & 4 & 5 & 2 & 3 & 0 \\
\color{blue}\mathbf{8} & \color{blue}\mathbf{4} & 
	\color{blue}\mathbf{5} & 0 & 2 & 6 & 1 & 7 & 3 \\
\color{blue}\mathbf{0} & \color{blue}\mathbf{8} & 
	\color{blue}\mathbf{3} & 4 & 5 & 7 & 6 & 1 & 2 \\
\color{blue}\mathbf{5} & \color{blue}\mathbf{2} & 
	\color{blue}\mathbf{6} & 3 & 1 & 8 & 4 & 0 & 7          
          
        \end{array}\right]\\
        \arraycolsep=2.5pt
        \left[ \begin{array}
          {ccccccccc}
4 & 2 & 7 & 8 & 0 & 1 & 3 & 5 & 6 \\
6 & 8 & 2 & 7 & 1 & 5 & 4 & 0 & 3 \\
5 & 6 & 4 & 3 & 2 & 8 & 1 & 7 & 0 \\
2 & 1 & 0 & 5 & 3 & 4 & 7 & 6 & 8 \\
8 & 5 & 1 & 0 & 4 & 6 & 2 & 3 & 7 \\
3 & 7 & 8 & 2 & 5 & 0 & 6 & 1 & 4 \\
7 & 0 & 5 & 1 & 6 & 3 & 8 & 4 & 2 \\
1 & 4 & 3 & 6 & 7 & 2 & 0 & 8 & 5 \\
0 & 3 & 6 & 4 & 8 & 7 & 5 & 2 & 1

		\end{array}\right]\\
        
        \arraycolsep=2.5pt
        \left[ \begin{array}
          {ccccccccc}

3 & 5 & 7 & 8 & 4 & 6 & 0 & 1 & 2 \\
4 & 0 & 2 & 7 & 6 & 3 & 1 & 5 & 8 \\
1 & 7 & 4 & 3 & 5 & 0 & 2 & 8 & 6 \\
7 & 6 & 0 & 5 & 2 & 8 & 3 & 4 & 1 \\
2 & 3 & 1 & 0 & 8 & 7 & 4 & 6 & 5 \\
6 & 1 & 8 & 2 & 3 & 4 & 5 & 0 & 7 \\
8 & 4 & 5 & 1 & 7 & 2 & 6 & 3 & 0 \\
0 & 8 & 3 & 6 & 1 & 5 & 7 & 2 & 4 \\
5 & 2 & 6 & 4 & 0 & 1 & 8 & 7 & 3

        \end{array}\right]\\
        \arraycolsep=2.5pt
        \left[ \begin{array}
          {ccccccccc}
          
0 & 4 & 2 & 5 & 6 & 1 & 7 & 3 & 8 \\
1 & 6 & 8 & 0 & 3 & 5 & 2 & 4 & 7 \\
2 & 5 & 6 & 7 & 0 & 8 & 4 & 1 & 3 \\
3 & 2 & 1 & 6 & 8 & 4 & 0 & 7 & 5 \\
4 & 8 & 5 & 3 & 7 & 6 & 1 & 2 & 0 \\
5 & 3 & 7 & 1 & 4 & 0 & 8 & 6 & 2 \\
6 & 7 & 0 & 4 & 2 & 3 & 5 & 8 & 1 \\
7 & 1 & 4 & 8 & 5 & 2 & 3 & 0 & 6 \\
8 & 0 & 3 & 2 & 1 & 7 & 6 & 5 & 4
          
        \end{array}\right]
        
      \end{aligned}$}
      
        & {$\!
        \begin{aligned}
        \arraycolsep=2.5pt
        
        \left[\begin{array}
         {cccccccccccc}

\color{red}\mathbf{9} & \color{red}\mathbf{10} & \color{red}\mathbf{11} & 3 & 4 & 5 & 6 & 7 & 8 & \color{blue}\mathbf{0} & \color{blue}\mathbf{1} & \color{blue}\mathbf{2}  \\

\color{red}\mathbf{9} & \color{red}\mathbf{10} & \color{red}\mathbf{11} & 4 & 6 & 0 & 3 & 2 & 7 & \color{blue}\mathbf{1} & \color{blue}\mathbf{5} & \color{blue}\mathbf{8}  \\

\color{red}\mathbf{9} & \color{red}\mathbf{10} & \color{red}\mathbf{11} & 1 & 5 & 7 & 0 & 4 & 3 & \color{blue}\mathbf{2} & \color{blue}\mathbf{8} & \color{blue}\mathbf{6}  \\

\color{red}\mathbf{9} & \color{red}\mathbf{10} & \color{red}\mathbf{11} & 7 & 2 & 6 & 8 & 0 & 5 & \color{blue}\mathbf{3} & \color{blue}\mathbf{4} & \color{blue}\mathbf{1}  \\

\color{red}\mathbf{9} & \color{red}\mathbf{10} & \color{red}\mathbf{11} & 2 & 8 & 3 & 7 & 1 & 0 & \color{blue}\mathbf{4} & \color{blue}\mathbf{6} & \color{blue}\mathbf{5}  \\

\color{red}\mathbf{9} & \color{red}\mathbf{10} & \color{red}\mathbf{11} & 6 & 3 & 1 & 4 & 8 & 2 & \color{blue}\mathbf{5} & \color{blue}\mathbf{0} & \color{blue}\mathbf{7}  \\

\color{red}\mathbf{9} & \color{red}\mathbf{10} & \color{red}\mathbf{11} & 8 & 7 & 4 & 2 & 5 & 1 & \color{blue}\mathbf{6} & \color{blue}\mathbf{3} & \color{blue}\mathbf{0}  \\

\color{red}\mathbf{9} & \color{red}\mathbf{10} & \color{red}\mathbf{11} & 0 & 1 & 8 & 5 & 3 & 6 & \color{blue}\mathbf{7} & \color{blue}\mathbf{2} & \color{blue}\mathbf{4}  \\

\color{red}\mathbf{9} & \color{red}\mathbf{10} & \color{red}\mathbf{11} & 5 & 0 & 2 & 1 & 6 & 4 & \color{blue}\mathbf{8} & \color{blue}\mathbf{7} & \color{blue}\mathbf{3}
               
        \end{array}\right]\\

	\arraycolsep=2.5pt
        \left[\begin{array}
          {cccccccccccc}

\color{red}\mathbf{10} & \color{red}\mathbf{11} & \color{red}\mathbf{9} & 7 & 5 & 8 & 2 & 4 & 0 & \color{blue}\mathbf{1} & \color{blue}\mathbf{3} & \color{blue}\mathbf{6}  \\

\color{red}\mathbf{10} & \color{red}\mathbf{11} & \color{red}\mathbf{9} & 2 & 0 & 7 & 8 & 6 & 1 & \color{blue}\mathbf{5} & \color{blue}\mathbf{4} & \color{blue}\mathbf{3}  \\

\color{red}\mathbf{10} & \color{red}\mathbf{11} & \color{red}\mathbf{9} & 4 & 7 & 3 & 6 & 5 & 2 & \color{blue}\mathbf{8} & \color{blue}\mathbf{1} & \color{blue}\mathbf{0}  \\

\color{red}\mathbf{10} & \color{red}\mathbf{11} & \color{red}\mathbf{9} & 0 & 6 & 5 & 1 & 2 & 3 & \color{blue}\mathbf{4} & \color{blue}\mathbf{7} & \color{blue}\mathbf{8}  \\

\color{red}\mathbf{10} & \color{red}\mathbf{11} & \color{red}\mathbf{9} & 1 & 3 & 0 & 5 & 8 & 4 & \color{blue}\mathbf{6} & \color{blue}\mathbf{2} & \color{blue}\mathbf{7}  \\

\color{red}\mathbf{10} & \color{red}\mathbf{11} & \color{red}\mathbf{9} & 8 & 1 & 2 & 7 & 3 & 5 & \color{blue}\mathbf{0} & \color{blue}\mathbf{6} & \color{blue}\mathbf{4}  \\

\color{red}\mathbf{10} & \color{red}\mathbf{11} & \color{red}\mathbf{9} & 5 & 4 & 1 & 0 & 7 & 6 & \color{blue}\mathbf{3} & \color{blue}\mathbf{8} & \color{blue}\mathbf{2}  \\

\color{red}\mathbf{10} & \color{red}\mathbf{11} & \color{red}\mathbf{9} & 3 & 8 & 6 & 4 & 1 & 7 & \color{blue}\mathbf{2} & \color{blue}\mathbf{0} & \color{blue}\mathbf{5}  \\

\color{red}\mathbf{10} & \color{red}\mathbf{11} & \color{red}\mathbf{9} & 6 & 2 & 4 & 3 & 0 & 8 & \color{blue}\mathbf{7} & \color{blue}\mathbf{5} & \color{blue}\mathbf{1}

        \end{array}\right]\\
        \arraycolsep=2.5pt
        \left[\begin{array}
          {cccccccccccc}
          
\color{red}\mathbf{11} & \color{red}\mathbf{9} & \color{red}\mathbf{10} & 2 & 6 & 0 & 8 & 4 & 1 & \color{blue}\mathbf{3} & \color{blue}\mathbf{5} & \color{blue}\mathbf{7}  \\

\color{red}\mathbf{11} & \color{red}\mathbf{9} & \color{red}\mathbf{10} & 8 & 3 & 1 & 7 & 6 & 5 & \color{blue}\mathbf{4} & \color{blue}\mathbf{0} & \color{blue}\mathbf{2}  \\

\color{red}\mathbf{11} & \color{red}\mathbf{9} & \color{red}\mathbf{10} & 6 & 0 & 2 & 3 & 5 & 8 & \color{blue}\mathbf{1} & \color{blue}\mathbf{7} & \color{blue}\mathbf{4}  \\

\color{red}\mathbf{11} & \color{red}\mathbf{9} & \color{red}\mathbf{10} & 1 & 8 & 3 & 5 & 2 & 4 & \color{blue}\mathbf{7} & \color{blue}\mathbf{6} & \color{blue}\mathbf{0}  \\

\color{red}\mathbf{11} & \color{red}\mathbf{9} & \color{red}\mathbf{10} & 5 & 7 & 4 & 0 & 8 & 6 & \color{blue}\mathbf{2} & \color{blue}\mathbf{3} & \color{blue}\mathbf{1}  \\

\color{red}\mathbf{11} & \color{red}\mathbf{9} & \color{red}\mathbf{10} & 7 & 4 & 5 & 2 & 3 & 0 & \color{blue}\mathbf{6} & \color{blue}\mathbf{1} & \color{blue}\mathbf{8}  \\

\color{red}\mathbf{11} & \color{red}\mathbf{9} & \color{red}\mathbf{10} & 0 & 2 & 6 & 1 & 7 & 3 & \color{blue}\mathbf{8} & \color{blue}\mathbf{4} & \color{blue}\mathbf{5}  \\

\color{red}\mathbf{11} & \color{red}\mathbf{9} & \color{red}\mathbf{10} & 4 & 5 & 7 & 6 & 1 & 2 & \color{blue}\mathbf{0} & \color{blue}\mathbf{8} & \color{blue}\mathbf{3}  \\

\color{red}\mathbf{11} & \color{red}\mathbf{9} & \color{red}\mathbf{10} & 3 & 1 & 8 & 4 & 0 & 7 & \color{blue}\mathbf{5} & \color{blue}\mathbf{2} & \color{blue}\mathbf{6}          
                    
        \end{array}\right]\\
        
        \arraycolsep=2.5pt
        \left[ \begin{array}
          {cccccccccccc}
          
4 & 2 & 7 & 8 & 0 & 1 & 3 & 5 & 6 & \color{red}\mathbf{9} & \color{red}\mathbf{10} & \color{red}\mathbf{11} \\
6 & 8 & 2 & 7 & 1 & 5 & 4 & 0 & 3 & \color{red}\mathbf{9} & \color{red}\mathbf{10} & \color{red}\mathbf{11} \\
5 & 6 & 4 & 3 & 2 & 8 & 1 & 7 & 0 & \color{red}\mathbf{9} & \color{red}\mathbf{10} & \color{red}\mathbf{11} \\
2 & 1 & 0 & 5 & 3 & 4 & 7 & 6 & 8 & \color{red}\mathbf{9} & \color{red}\mathbf{10} & \color{red}\mathbf{11} \\
8 & 5 & 1 & 0 & 4 & 6 & 2 & 3 & 7 & \color{red}\mathbf{9} & \color{red}\mathbf{10} & \color{red}\mathbf{11} \\
3 & 7 & 8 & 2 & 5 & 0 & 6 & 1 & 4 & \color{red}\mathbf{9} & \color{red}\mathbf{10} & \color{red}\mathbf{11} \\
7 & 0 & 5 & 1 & 6 & 3 & 8 & 4 & 2 & \color{red}\mathbf{9} & \color{red}\mathbf{10} & \color{red}\mathbf{11} \\
1 & 4 & 3 & 6 & 7 & 2 & 0 & 8 & 5 & \color{red}\mathbf{9} & \color{red}\mathbf{10} & \color{red}\mathbf{11} \\
0 & 3 & 6 & 4 & 8 & 7 & 5 & 2 & 1 & \color{red}\mathbf{9} & \color{red}\mathbf{10} & \color{red}\mathbf{11}
        \end{array}\right]\\
        \arraycolsep=2.5pt
        \left[ \begin{array}
          {cccccccccccc}
          
3 & 5 & 7 & 8 & 4 & 6 & 0 & 1 & 2 & \color{red}\mathbf{10} & \color{red}\mathbf{11} & \color{red}\mathbf{9} \\
4 & 0 & 2 & 7 & 6 & 3 & 1 & 5 & 8 & \color{red}\mathbf{10} & \color{red}\mathbf{11} & \color{red}\mathbf{9} \\
1 & 7 & 4 & 3 & 5 & 0 & 2 & 8 & 6 & \color{red}\mathbf{10} & \color{red}\mathbf{11} & \color{red}\mathbf{9} \\
7 & 6 & 0 & 5 & 2 & 8 & 3 & 4 & 1 & \color{red}\mathbf{10} & \color{red}\mathbf{11} & \color{red}\mathbf{9} \\
2 & 3 & 1 & 0 & 8 & 7 & 4 & 6 & 5 & \color{red}\mathbf{10} & \color{red}\mathbf{11} & \color{red}\mathbf{9} \\
6 & 1 & 8 & 2 & 3 & 4 & 5 & 0 & 7 & \color{red}\mathbf{10} & \color{red}\mathbf{11} & \color{red}\mathbf{9} \\
8 & 4 & 5 & 1 & 7 & 2 & 6 & 3 & 0 & \color{red}\mathbf{10} & \color{red}\mathbf{11} & \color{red}\mathbf{9} \\
0 & 8 & 3 & 6 & 1 & 5 & 7 & 2 & 4 & \color{red}\mathbf{10} & \color{red}\mathbf{11} & \color{red}\mathbf{9} \\
5 & 2 & 6 & 4 & 0 & 1 & 8 & 7 & 3 & \color{red}\mathbf{10} & \color{red}\mathbf{11} & \color{red}\mathbf{9}

       	\end{array}\right]\\
        \arraycolsep=2.5pt
        \left[ \begin{array}
          {cccccccccccc}
           
0 & 4 & 2 & 5 & 6 & 1 & 7 & 3 & 8 & \color{red}\mathbf{11} & \color{red}\mathbf{9} & \color{red}\mathbf{10} \\
1 & 6 & 8 & 0 & 3 & 5 & 2 & 4 & 7 & \color{red}\mathbf{11} & \color{red}\mathbf{9} & \color{red}\mathbf{10} \\
2 & 5 & 6 & 7 & 0 & 8 & 4 & 1 & 3 & \color{red}\mathbf{11} & \color{red}\mathbf{9} & \color{red}\mathbf{10} \\
3 & 2 & 1 & 6 & 8 & 4 & 0 & 7 & 5 & \color{red}\mathbf{11} & \color{red}\mathbf{9} & \color{red}\mathbf{10} \\
4 & 8 & 5 & 3 & 7 & 6 & 1 & 2 & 0 & \color{red}\mathbf{11} & \color{red}\mathbf{9} & \color{red}\mathbf{10} \\
5 & 3 & 7 & 1 & 4 & 0 & 8 & 6 & 2 & \color{red}\mathbf{11} & \color{red}\mathbf{9} & \color{red}\mathbf{10} \\
6 & 7 & 0 & 4 & 2 & 3 & 5 & 8 & 1 & \color{red}\mathbf{11} & \color{red}\mathbf{9} & \color{red}\mathbf{10} \\
7 & 1 & 4 & 8 & 5 & 2 & 3 & 0 & 6 & \color{red}\mathbf{11} & \color{red}\mathbf{9} & \color{red}\mathbf{10} \\
8 & 0 & 3 & 2 & 1 & 7 & 6 & 5 & 4 & \color{red}\mathbf{11} & \color{red}\mathbf{9} & \color{red}\mathbf{10}
        
        \end{array}\right]
  \end{aligned}$}
\end{tabular}
\caption{An example of rudimentary parallel partition and extension, with $n=9, d=9, r=3$. The column on the left shows a PA $A$ consisting of six blocks of permutations on $Z_9$ with $hd(A) \geq 6$. The column on the right shows the new PA $A'$ on $Z_{12}$ with $hd(A') \geq 6$.}
\label{tb:par_pe_ex}
\end{table}

\clearpage

\subsection{General Parallel Partition with $r$ Symbols}

\label{general_parallel}
As described in Section \ref{rudi}, rudimentary parallel partition and extension with $r=2$  allows extension of at most $2r=4$ blocks. 
We describe a new technique, called {\em general parallel partition and extension with $r$ symbols}, that allows a larger number of blocks to be extended. 

We start with the simplest form of general parallel partition and extension, for $r=2$ symbols.
It expands on the simple partition and extension technique described in Section \ref{s:prev} by introducing an additional pair of partitions of $Z_n$, denoted by $\mathcal{R}$ and $\mathcal{S}$ in the description that follows.

Let $s$ be a positive integer, and let $M_1, M_2, \dots, M_s$ be an ordered list of $s$ pairwise disjoint PAs on $Z_n$. Let 
$\P=(P_1, P_2, \dots, P_s )$, 
$\Q=(Q_1, Q_2, \dots, Q_s )$,
$\R=(R_1, R_2, \dots, R_s )$, and 
$\S=(S_1, S_2, \dots, S_s )$, 
be four partitions of $Z_n$
such that, for all $i$, $P_i\cap R_i=\emptyset$ and $Q_i\cap S_i=\emptyset$.    
The sets $P_i$ and $R_i$ are sets of locations for replacing symbols in the PA $M_i$, and 
the sets $Q_i$ and $S_i$ are sets of symbols to be replaced.   
For each $i$, let $\text{2-}covered(M_i)$ be defined by
\[
\text{2-}covered(M_i)=\{\sigma \in M_i  ~|~ \exists p\in P_i, ~ \exists r \ne p \in R_i ~ (\sigma(p)\in Q_i , ~\sigma(r)\in S_i) \}.
\]

We say that a permutation $\sigma$ is $\text{2-}covered$ if $\sigma \in \text{2-}covered(M_i)$ for some $i$. 
In general, when $\sigma$ is $\text{2-}covered$, there may be multiple pairs $(p,r)\in P_i\times R_i$ 
such that $\sigma(p)\in Q_i$ and $\sigma(r)\in S_i$.
If so, arbitrarily designate one of these pairs 
to cover $\sigma$. We use the notation $(p,r)$ to refer to the designated pair.

The {\em parallel extension of $\sigma$ by the pair $(p,r)$}, denoted by $\ext(\sigma)=\sigma'$, 
is a permutation on $Z_{n+2}$ defined by 
\begin{align}
\ext(\sigma(x))=\sigma'(x)= 
\begin{cases}
   n  			& \text{if } x=p\\
   \sigma(p) 	& \text{if } x=n\\
   n+1  		& \text{if } x=r\\
   \sigma(r) 	& \text{if } x=n+1\\   
   \sigma(j) 	& \forall j,~(0\le j<n ~\land ~j\notin\{p,r\}).
\end{cases} 
\label{eq:2ext-sigma}
\end{align}
We will always extend  $\sigma$ at the designated pair of positions $(p,r)$ and refer to this 
new permutation as $\ext(\sigma)$ or $\sigma'$ interchangeably. Note that in order for a permutation $\sigma'$ to be included in the extended set of permutations on $n+2$ 
symbols, $\sigma$ must be 2-covered. 
In other words, $\sigma$ must have two of the named symbols in two of the named positions. 

For our construction, we include two additional PAs, $M_{s+1},M_{s+2}$, 
 for which there are no corresponding sets of positions or symbols. 
None of the permutations in $M_{s+1}$ or $M_{s+2}$ are in any of the sets $M_{i}~(1\le i\le s)$.  
In a manner similar to rudimentary parallel partition and extension, parallel partition and extension extends $M_{s+1}$ and $M_{s+2}$ by appending the two new symbols $n$ and $n+1$, to the end of each permutation. For 
$M_{s+1}$, the sequence $(n,n+1)$ is appended to the end of each permutation. Similarly, for $M_{s+2}$, the sequence $(n+1,n)$ is appended to the end of each permutation. 
Every permutation in $M_{s+1}$ and $M_{s+2}$ is used in the construction of our new PA. 
We create the list
$\M=(M_1, M_2, \dots, M_{s+1}, M_{s+2})$, which includes the extra sets $M_{s+1}$ and $M_{s+2}$. 

A partition system $\Pi=({\M,\P,\Q,\R,\S})$ is a {\em $(d,2)$-partition system} for $Z_n$ if it 
satisfies the following properties:
\begin{enumerate} [\indent(I)]
\item \label{(d,2) prop1} $\forall M_i\in{\mathcal M}, ~hd(M_i)\ge d$, and
\item \label{(d,2) prop2} $\forall i,j ~(1\le i<j\le s+2),~ hd(M_i,M_j)\ge d-2$.
\end{enumerate}

Parallel partition and extension uses sets $P_i,Q_i,R_i,$ and $S_i$ from the partitions $\P,\Q,\R,$
and $\S$, respectively, to modify the 2-covered permutations in $M_i$, for $1\le i\le s$, 
for the purpose of creating a new PA on $Z_{n+2}$ with Hamming distance $d$.  
Let $\Pi=(\M,\P,\Q,\R,S)$ be a $(d,2)$-partition system, where 
${\mathcal M}=(M_1, M_2, \dots, M_{s+2})$, for some $s$.  
We now show how parallel partition and extension operation creates a new permutation array 
$\ext(\Pi)$ on $Z_{n+2}$.
For all $i$ $(1\le i\le s)$, let $\ext(M_i)$ be the set of permutations defined by 
\[
\ext(M_i)=\{\ext(\sigma)~|~\sigma \in \text{2-}covered (M_i)\}.
\] 
For $M_{s+1}$, let $\ext(M_{s+1}) $ be the set of permutations on $Z_{n+2}$ defined by 
adding the symbols $n$ and $n+1$, in that order, to the end of every permutation of $M_{s+1}$. For $M_{s+2}$, let $\ext(M_{s+2})$ be the set of permutations on $Z_{n+2}$ defined by adding the symbols $n+1$ and $n$, in that order, to the end of every permutation of $M_{s+2}$.

Let $\ext(\Pi)$ be defined by 
\[
\ext(\Pi)=\bigcup_{i=1}^{s+2} ~\ext(M_i). 
\]
Note that 
\[
|\ext(\Pi)|=\sum_{i=1}^{s+2} |\ext(M_i)|.
\]

\begin{thm} \label{thm:parallel_pe_2}
Let $d$ be a positive integer,
let  $\Pi=(\M,\P,\Q,\R,\S)$ be a $(d,2)$-partition system for $Z_n$, 
with $\M=(M_1,M_2,\dots,M_{s+2})$ for some positive integer $s$.
Let $\ext(\Pi)$ be the PA on $Z_{n+2}$ created by parallel partition and extension.
Then, $hd(\ext(\Pi))\ge d$. 
\end{thm}

\begin{proof}
Our proof has three steps. We first use simple partition and extension to create a PA $ext(\Pi')$, on $Z_{n+1}$, that exhibits $hd(ext(\Pi'))\ge d-1$. Next, using simple partition and extension again, we create a PA $ext(\Pi'')$, on $Z_{n+2}$, that exhibits $hd(ext(\Pi''))\ge d$. Finally, we show that the PA $\ext(\Pi)=  ext(\Pi'') \cup \ext(M_{s+1}) \cup \ext(M_{s+2})$  exhibits $hd(\ext(\Pi))\ge d$.

Consider $\M'=(M_1, M_2, \dots, M_s)$.
First, observe that $\Pi'=(\M',\P,\R)$ can be viewed as a distance-$(d-1)$ partition system for $Z_n$ 
since $hd(M_i)\ge d\ge d-1$ for all $i$, ($1\le i\le s$) and $hd(M_i,M_j)\ge d-2$ for all $i,j$, ($1\le i<j\le s$). Simple partition and extension on $\Pi'$ results in the PA $ext(\Pi')$ on $Z_{n+1}$.
By Theorem \ref{th:simple_pe}, $hd(ext(\Pi'))\ge d-1$. 
In particular, for all $i,j$ ($1\le i,j\le s,~ i\neq j$), $hd(ext(M_i),ext(M_j))\ge d-1$.

Notice that, for all $i$ ($1\le i\le s$), $hd(ext(M_i))\ge d$ since $hd(M_i)\ge d$.
(As shown in \cite{bms17}, this follows from case 1 in the proof of Theorem \ref{th:simple_pe}.
For two permutations $\sigma$ and $\tau$ from the same set $M_i$, 
at most one new agreement appears between $ext(\sigma)$ and $ext(\tau)$. 
Since $ext(\sigma)$ and $ext(\tau)$ are in $Z_{n+1}$, $hd(ext(\sigma),ext(\tau))= hd(\sigma,\tau) \geq d$. 
See \cite{bms17} for the full proof of Theorem \ref{th:simple_pe}.)

Let $\M''=(ext(M_1), ext(M_2), \dots, ext(M_s))$.
Then $\Pi''=(\M'',\R,\S)$ is a distance-$d$ partition system for $Z_{n+1}$.
Simple partition and extension on $\Pi''$ results in the PA $ext(\Pi'')$ on $Z_{n+2}$.
By Theorem \ref{th:simple_pe}, $hd(ext(\Pi''))\ge d$.

By assumption, $\Pi$ is a $(d,2)$-partition system, so, by property \ref{(d,2) prop1} of $(d,2)$ partition systems, $hd(M_{s+1})\ge d$ and $hd(M_{s+2})\ge d$.
By definition, every permutation $\tau'$ in $\ext(M_{s+1})$ is built from a permutation  $\tau$ in $M_{s+1}$ by appending the sequence $(n,n+1)$ to the end. 
This increases the length of each permutation by 2, and number of agreements between every pair of permutations in $\ext(M_{s+1})$ by 2.
So $hd(\ext(M_{s+1}))=n+2-((n-d)+2) \ge d$.
Similar reasoning applies to every permutation in $\ext(M_{s+2})$ using the appended sequence $(n+1,n)$, so $hd(\ext(M_{s+2}))\ge d$. 
Let $\tau' \in \ext(M_{s+1})$ and $\rho' \in \ext(M_{s+2})$ be arbitrary permutations.
The appended sequences $(n,n+1)$ and $(n+1,n)$ create no new agreements between $\tau'$ and $\rho'$.  
By property \ref{(d,2) prop2} of $(d,2)$ partition systems, $\forall i,j ~(1 \le i<j \le s+2),~ hd(M_{i},M_{j})\ge d-2$.
In particular, $hd(M_{s+1},M_{s+2})\ge d-2$.
So it follows that $hd(\ext(M_{s+1}),\ext(M_{s+2}))\ge n+2-(n-(d-2))=d$.

To see that $hd(ext(\Pi''),\ext(M_{s+1}))\ge d$, let $\sigma'' \in ext(\Pi'')$.
Extending the original permutation $\sigma$ to create $\sigma''$ merely replaces designated symbols in designated positions with the symbols $n$ and $n+1$, and moves the displaced symbols to positions $n$ and $n+1$, respectively. 
On the other hand, for any permutation $\tau' \in \ext(M_{s+1})$, the symbols $n$ and $n+1$ are in positions $n$ and $n+1$.
In both cases, no other symbols are moved. So the symbols $n$ and $n+1$ in $\sigma''$ are not in the same locations as they are in $\tau'$ and neither are the displaced symbols. That is, no new agreements are created.
Hence, $hd(ext(\Pi''),\ext(M_{s+1}))\ge n+2-(n-(d-2))=d$.
Similarly, $hd(ext(\Pi''),\ext(M_{s+2}))\ge n+2-(n-(d-2))=d$.

Finally, observe that  $\ext(\Pi)=  ext(\Pi'') ~\cup~ \ext(M_{s+1}) ~\cup~ \ext(M_{s+2})$. 
We showed above that the pairwise Hamming distance between all PAs in $\ext(\Pi)$ is at least $d$, so it follows that $hd(\ext(\Pi)) \ge d$.
\end{proof}

\begin{changemargin}{1cm}{1cm}
\begin{example}
\label{example_parallel_pe}
This example illustrates the use of Theorem \ref{thm:parallel_pe_2} to construct a PA for $n=40$ and $d=34$.
We start with $PGL(2,37)$ is a PA on $Z_{38}$. 
It contains $38\cdot 37\cdot 36=50,616$ permutations with Hamming distance at least 36, giving $M(38,36) \ge 50,616 $. 
Using the coset method \cite{bereg2015constructing}, we found five cosets of $PGL(2,37)$ in $S_{38}$, with Hamming distance 34 from $PGL(2,37)$ (see Table \ref{tb:coset}). 
The cosets are defined by the coset representatives $\alpha,\beta,\gamma,\delta$ and $\theta$: 
{\scriptsize
\arraycolsep=.18em
\[
\begin{matrix} 
\alpha=27 & 12 & 30 & 25 & 15 & 37 & 35 & 22 & 29 & 36 & 10 & 1 & 13 & 33 & 24 & 3 & 28 & 16 & 26 & 8 & 19 & 17 & 23 & 0 & 11 & 34 & 20 & 5 & 31 & 6 & 21 & 14 & 18 & 32 & 7 & 9 & 2 & 4\\
\beta=16 & 22 & 35 & 6 & 4 & 30 & 37 & 26 & 23 & 11 & 0 & 20 & 18 & 24 & 8 & 7 & 15 & 13 & 1 & 29 & 36 & 27 & 17 & 33 & 3 & 9 & 10 & 14 & 32 & 25 & 12 & 19 & 28 & 21 & 2 & 31 & 5 & 34\\
\gamma=12 & 26 & 21 & 32 & 37 & 24 & 2 & 9 & 23 & 27 & 0 & 30 & 18 & 16 & 20 & 11 & 6 & 34 & 33 & 29 & 15 & 22 & 5 & 10 & 17 & 4 & 35 & 13 & 28 & 1 & 14 & 25 & 7 & 36 & 19 & 3 & 31 & 8\\
\delta=17 & 28 & 22 & 37 & 26 & 9 & 8 & 12 & 18 & 4 & 32 & 33 & 31 & 5 & 2 & 1 & 34 & 29 & 0 & 3 & 21 & 6 & 10 & 16 & 23 & 36 & 20 & 15 & 14 & 35 & 11 & 30 & 19 & 24 & 25 & 7 & 13 & 27\\
\theta=9 & 30 & 12 & 6 & 36 & 13 & 31 & 11 & 1 & 17 & 27 & 26 & 5 & 24 & 14 & 35 & 25 & 10 & 23 & 7 & 34 & 18 & 20 & 2 & 16 & 0 & 8 & 19 & 29 & 15 & 37 & 33 & 4 & 21 & 22 & 32 & 28 & 3

\end{matrix}
\]
} 
Let $\M=\{M_1, M_2, M_3, M_4, M_5, M_6 \}$ where 
\begin{align*}
&M_1=PGL(2,37) & M_2=\alpha M_1 & & M_3=\beta M_1 & &M_4=\gamma M_1 & & M_5=\delta M_1 & & M_6=\theta M_1. 
\end{align*} 
Note that for all $i,j,~(1 \le i<j\le 6),~ hd(M_i)=36$ and $hd(M_i,M_j)\ge 34$.

\bigskip
\noindent Let $X=\{X_1,X_2,X_3,X_4\}$ be the partition of $Z_{38}$ given by
\vspace{-.2 cm}
\begin{align*}
  &X_1=\{0, 4, 8, 13, 19, 22, 26, 30, 35\} & &X_3=\{2, 6, 10, 12, 16, 21, 24, 28, 33, 37\} \\
  &X_2=\{1, 5, 9, 15, 18, 23, 27, 31, 34\} & &X_4=\{3, 7, 11, 14, 17, 20, 25, 29, 32, 36\}.
\end{align*}
The two partitions of positions, $\P$ and $\R$, are based on X. 
That is, $\mathcal{P}=\{P_1,P_2,P_3,P_4\}$, where  $P_1=X_1, P_2=X_2, P_3=X_3,$ and $P_4=X_4$
and $\mathcal{R}=\{R_1,R_2,R_3,R_4\}$, where $R_1=X_2,R_2=X_3, R_3=X_4,$ and $R_4=X_1$. 

\bigskip
\noindent Let $Y=\{Y_1,Y_2,Y_3,Y_4\}$ be the partition of $Z_{38}$ given by
\vspace{-.2 cm}
\begin{align*}
  &Y_1=\{0, 1, 2, 3, 4, 5, 6, 7, 8, 9\} & &Y_3=\{20, 21, 22, 23, 24, 25, 26, 27, 28\} \\
  &Y_2=\{10, 11, 12, 13, 14, 15, 16, 17, 18, 19\} & &Y_4=\{29, 30, 31, 32, 33, 34, 35, 36, 37\}.
\end{align*}
The two partitions of symbols, $\Q$ and $\S$, are based on Y. 
That is, $\mathcal{Q}=\{Q_1,Q_2,Q_3,Q_4\}$ where $Q_1=Y_1,Q_2=Y_2,Q_3=Y_3,Q_4=Y_4$ and $\mathcal{S}=\{S_1,S_2,S_3,S_4\}$ where $S_1=Y_2,S_2=Y_3,S_3=Y_4,S_4=Y_1$.

\bigskip
\noindent Let $\Pi=({\M,\P,\Q,\R,\S})$. 
It can be verified that $\Pi$ is a $(d,2)$-partition system for $Z_{38}$ where ${d=34}$. Parallel partition and extension on $\Pi$ results in $\ext(\Pi)$, where ${|\ext(\Pi)|=287,437}$.
Theorem \ref{thm:parallel_pe_2} for $n=38$ and $d=34$ implies $M(40,34)\ge 287,437$ which is a new lower bound. See Table \ref{tb:new_parallel_pe}.
\end{example}
\end{changemargin}
\bigskip

Theorem \ref{thm:parallel_pe_2} applies to general parallel partition and extension using $r=2$ symbols. 
This result can be generalized to arbitrary $r$ provided that a sufficient number of blocks with appropriate Hamming distance properties can be found, along with a corresponding number of partitions of positions and symbols. 
Table \ref{tb:new_parallel_pe} shows new bounds obtained using  parallel partition and extension (Theorems \ref{thm:rudi} and \ref{thm:parallel_pe_2}). 

The general parallel partition and extension technique does not put restrictions on the partitions of positions $\P,\R, ...$, and partitions of symbols $\Q,\S,...$, making the search space for good partitions very large. 
Because of this, we have experimented with several ways of creating partitions. For example, given a partition of positions $\P=\{P_0,P_1,...P_{k-1}\}$, a family of partitions $\{\P_i\}$ can be derived from $\P$ as follows.
For all $i,~(i\le 0 < k)$, define $\P_i$, the $i^{th}$  partition of positions, to be $\P_i=\{ P_{(i+j)\pmod k},~\forall (0\le j<k)\}$.
Using this notation, the partitions $\P$ and $\R$ of Example \ref{example_parallel_pe} are correspond to $\P_0$ and $\P_1$.  
In other words, $\P_1$ is obtained by a cyclic shift of the sets in $\P_0$.
In this way, each partition $\P_i$ comprises a different partition of the set of positions. 
Define a similar family of partitions of symbols $\{\Q_i\}$ using a partition of symbols $\Q=\{Q_0,Q_1,...Q_{k-1}\}$ as a starting point. 
Clearly, each pair of partitions $(\P_i,\Q_i)$ satisfies the conditions of the parallel partition and extension technique. To create the initial partitions $\P$ and $\Q$, we have used several techniques, including a greedy technique and a technique based on Integer Linear Programming. These are described in Sections \ref{ss:greedy_partition_sel} and \ref{ss:ilp_partition_sel}.

Results obtained by parallel partition and extension can be compared with results from the \textit{coset method} \cite{bereg2015constructing} and the \textit{contraction method} \cite{bereg2015constructing}. 
The coset method starts with a group $X$ exhibiting $M(n,d')$, for some $d'>d$ and searches for cosets of $X$ at Hamming distance $d$. 
The PA $A$, formed from $X$ together with its cosets, exhibits Hamming distance $d$. 
If $X$ is a good PA for $M(n,d')$, the PA $A$ could represent a new lower bound for $M(n,d)$.
The operation of contraction on a PA $Y$ on $Z_{n+1}$ with Hamming distance $d+1$ results in new PA $Y'$ on $Z_n$. As with the coset method, if $Y$ is a good PA for $M(n+1,d)$, $Y'$ could exhibit a new lower bound for either $M(n,d-2)$ or $M(n,d-3)$, depending on conditions described in \cite{bereg2015constructing}. 

To be competitive, the groups that serve as the starting point for any of these methods must be large. We have used $AGL(1,q)$ and $PGL(2,r)$ for various powers of primes $q$ and $r$.
The coset method and the contraction method are quite fruitful, but there are instances where parallel partition and extension gives better results for $M(n,d)$.

We have also experimented with several methods for generating blocks of permutations with a desired Hamming distance. 
For example, to search for new PAs that exhibit improved lower bounds for $M(n,d)$, one technique
looks for cosets at Hamming distance d from a group $G$ on $Z_{n-r}$ that exhibits $M(n-r,d')$, where $d'>d$. Let $\M$ consist of $G$ and the cosets. 
Using parallel partition and extension, the permutations in $\M$ are extended by $r$ symbols to create a new PA on $Z_n$ exhibiting $M(n,d)$. Our coset search techniques are discussed in Section \ref{ss:coset_search}.

\begin{table}[htb]

\centering
\vspace*{4mm}
\begin{tabular}{|r r r r c|}
\hline
\bfseries $n$ & \bfseries $d$ & $r$ & \bfseries NEW & \bfseries Origin of Blocks (see Table \ref{tb:coset}) \\
\hline
\hline
30 & 26 & 2 & $58,968_R$ & $P\Gamma L(2,27)$ and 2 cosets \\
40 & 34 & 2 & $287,437_P$ & $PGL(2,37)$ and 2 cosets (see $M(38,32)$) \\
44 & 38 & 2 & $397,198_P$  & $PGL(2,41)$ and 2 cosets (see $M(42,36)$) \\
45 & 39 & 3 & $413,280_R$  & $PGL(2,41)$ and 3 cosets (see $M(42,36)$) \\
46 & 39 & 4 & $551,040_R$  & $PGL(2,41)$ and 4 cosets (see $M(42,35)$) \\
52 & 46 & 2 & $470,397_R$ & $PGL(2,49)$ and 2 cosets (see $M(50,44)$) \\
53 & 47 & 3 & $470,400_R$ & $PGL(2,49)$ and 3 cosets (see $M(50,44)$) \\
56 & 50 & 2 & $446,472_R$ & $PGL(2,53)$ and 2 cosets (see $M(54,48)$) \\
70 & 63 & 2 & $1,503,462_P$ & $PGL(2,67)$ and 2 cosets (see $M(68,61)$) \\
\hline
\end{tabular}
\caption{$M(n,d)$ lower bounds obtained using {\em parallel partition and extension} (Theorem \ref{thm:rudi} and \ref{thm:parallel_pe_2}). The blocks used by these theorems were obtained by the coset method \cite{bereg2015constructing} (see Table \ref{tb:coset}). 
Columns: $r$ denotes the number of new symbols, {\bf NEW} denotes the new new bound. 
New bounds computed using rudimentary parallel partition and extension (Theorem \ref{thm:rudi}) and general parallel partition and extension (Theorem \ref{thm:parallel_pe_2}) are denoted with a subscript $R$  and $P$, respectively.
}
\label{tb:new_parallel_pe}
\end{table}

\bigskip

\section{Partition and Extension of Modified Kronecker Product}
\label{s:kron}

Kronecker product is a well known operation in linear algebra, combinatorics, and other areas of mathematics \cite{henderson1983history,holmquist1985direct}. 
A modification of the Kronecker product operation on PAs can be used to create larger PAs suitable for simple partition and extension. 

Let $X$ and $Y$ be PAs defined by $X=\{\alpha_1,\alpha_2,\dots,\alpha_l\}$ where each $\alpha_i$ is a permutation on $l$ symbols,
and $Y=\{\beta_1,\beta_2,\dots,\beta_m\}$ where each $\beta_i$ is a permutation on $m$ symbols. 
The notation $\alpha_i(j)$ denotes the symbol in permutation $\alpha_i$ at position $j$.
Let $(\alpha_i(j),Y)$ denote a modified copy of the PA $Y$ such that each symbol in each permutation of $Y$ has an offset $m \cdot \alpha_i(j)$ added to it. 
Clearly $|(\alpha_i(j),Y)|=|Y|$. Moreover, like $Y$, $(\alpha_i(j),Y)$ is a PA on $m$ symbols, however, the symbol set of $(\alpha_i(j),Y)$ is offset by the value $m \cdot \alpha_i(j)$. Hence the PAs $Y$ and $(\alpha_i(j),Y)$ have no symbols in common.

Let $(X\otimes Y)_i$ be the PA defined by $(X\otimes Y)_i=[(\alpha_i(0),Y), (\alpha_i(1),Y), \dots, (\alpha_i(l-1),Y)]$. 
That is, if $\beta_r$ is the permutation in $Y$, there is a corresponding permutation $\gamma$ on $lm$ symbols in $(X\otimes Y)_i$ of the form $\gamma=
(m\cdot\alpha_i(0)+\beta_r(0)),\dots, (m\cdot\alpha_i(0)+\beta_r(m-1)),
(m\cdot\alpha_i(1)+\beta_r(0)),\dots, (m\cdot\alpha_i(1)+\beta_r(m-1)), 
\dots,
(m\cdot\alpha_i(l-1)+\beta_r(0)),\dots, 
(m\cdot\alpha_i(l-1)+\beta_r(m-1)).$
In other words, $\gamma$ can be viewed as the concatenation of $l$ copies of $\beta_r$ with an appropriate offset added to the symbols in each copy.
The  offsets ensure that each of the $|Y|$ rows in the sub-array $(X\otimes Y)_i$ is a permutation on the $lm$ symbols $\{0,1,2\dots lm-1\}$. 

Define the modified Kronecker product \cite{bmms-kp-17} of PAs $X$ and $Y$, denoted by $(X\otimes Y)$, to be the PA on $lm$ symbols defined by $(X\otimes Y)=\bigcup \limits_{i=1}^l (X\otimes Y)_i$. This is illustrated in Figure \ref{kronfig}.

\begin{figure}[hbt]
\centerline{\includegraphics{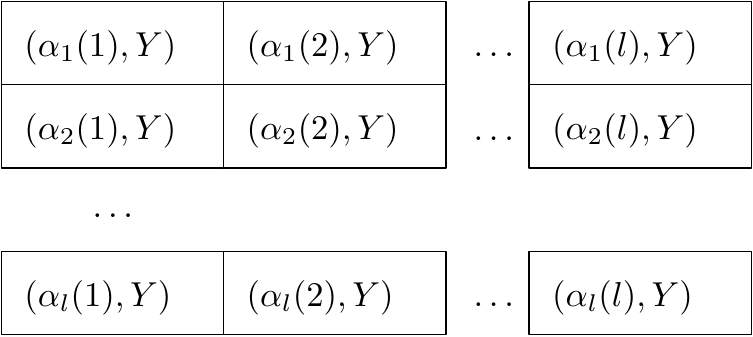}}
\caption{The PA $(X\otimes Y)$, the modified Kronecker product of PA's $X$ and $Y$.} 
\label{kronfig}
\end{figure}

Define the {\em block decomposition} of a PA $A$ on $n$ symbols as a collection of sub-arrays (\textit{i.e., blocks}), say  $A^{(1)},A^{(2)},\dots,A^{(m)}$, 
such that for all $i~(1\le i\le m)$, $hd(A^{(i)})=n$.
A detailed discussion of block decomposition appears in \cite{bmms-kp-17}, along with several examples using $AGL(1,q)$ and $PGL(2,q)$, where $q$ is a prime or a prime power.
We use block decompositions of PAs and the modified Kronecker product to produce new PAs, which in some cases give new lower bounds for $M(n+1,n)$. Corollaries \ref{cor:kron1} and \ref{cor:kron2} below describe our results. 
Our block decompositions have a property that the blocks are {\em full}, \ie $|A^{(i)}|=n$.  
We need two lemmas describing properties of PAs produced by modified Kronecker product to establish Corollaries \ref{cor:kron1} and \ref{cor:kron2}.

\begin{lemma} [\cite{bmms-kp-17}]
\label{l:kron2}
Let $A^{(1)},A^{(2)},\dots,A^{(k)}$ be a block decomposition of a PA $A$ on $l$ symbols with $hd(A)=l-a$ 
Let $B^{(1)},B^{(2)},\dots B^{(k)}$ be a block decomposition of PA $B$ on $m$ symbols with $hd(B)=m-b$. 
Let $M_i=A^{(i)}\otimes B^{(i)}$
Then $$hd(\bigcup \limits_{i=1}^k M_i )=lm-ab.$$
\end{lemma}

\begin{lemma} 
\label{l:kron3}
Let $A^{(1)},A^{(2)},\dots,A^{(k)}$ be a block decomposition of a PA $A$ on $l$ symbols with $hd(A)=l-1$. 
Let $B^{(1)},B^{(2)},\dots B^{(k)}$ be a block decomposition of PA $B$ on $m$ symbols with $hd(B)=m-1$. 
Then $M(n+1,n)\ge kn$, where $n=lm$.
\end{lemma}

\begin{proof}
First, we set $\M=\{M_1,M_2,\dots,M_k\}$ where 
for all $i, ~(i=1,2,\dots,k), M_i=A^{(i)}\otimes B^{(i)}$.
That is, $M_i$ is the modified Kronecker product of the blocks $A^{(i)}$ and $B^{(i)}$. 
The PA $M_i$ can be viewed as an $l\times l$ table of blocks.
In particular, the columns of this table are columns of blocks, and the rows of the table are rows of blocks. We will refer to the rows and columns as {\em block rows} and {\em block columns}, respectively.
Let $C_1,C_2,\dots,C_l$ be the block columns of the table. 
For each block column $C_j$, $(j=1,2,\dots,l)$ we select the $(i-1)^{st}$ position in $C_j$, keeping in mind that positions are numbered starting at 0.
Let $P_i$ be the set of selected positions.
That is, $P_i=\{i-1, (i-1)+l, (i-1)+2l, \dots,(i-1)+kl \}$.
We choose the symbols for $Q_i$ as $0,1,\dots,m-1$ with added offset $(i-1)m$.
That is, $Q_i=\{0+(i-1)m,1+(i-1)m,\dots,(m-1)+(i-1)m  \}$.
Note that each block row of the table contains a block column such that all symbols in it have offset $(i-1)m$.
Therefore all permutations in this block row are covered.
The lemma follows since all $klm$ permutations of the modified Kronecker product are covered.
\end{proof}

\bigskip
\begin{cor} 
\label{cor:kron1}
Let $p$ and $q$ be prime powers.
Let $n=pq$ and $k = \min \{p-1,q-1\}$. Then $M(n+1,n)\geq kn$.
\end{cor}

\begin{proof}
It follows from Lemma \ref{l:kron3} if we take the affine general linear groups
$A=AGL(1,p)$ and $B=AGL(1,q)$.
\end{proof}

\bigskip

\begin{cor} 
\label{cor:kron2}
Let $n\ge 2$ and $m\ge 2$ be integers. 
Let $N_n$ be the maximum number of MOLS of order $n$.
Let $k=\min\{N_n,N_m\}$.
Then $M(nm+1,nm)\geq knm$.  
\end{cor}

\begin{proof}
Colbourn, Kl{\o}ve and Ling \cite{colbourn2004permutation} proved that a set of $k$ MOLS of order $n$ can be transformed into a permutation array $A$ of size $kn$ on $Z_n$.
Each Latin square $C_s$ is transformed into a block $D_s$ of $n$ permutations with pairwise Hamming distance $n$. 
The transformation changes triples $(i,j,k) \in C_s$  to triples $(k,j,i) \in D_s$. 
In other words, for all $i,j,k \in Z_n$ the symbol $k$ in row $i$ and column $j$ in the Latin square $C_s$ becomes the symbol $i$ in row $k$ and column $j$ in the block $D_s$.  

Suppose there are k MOLS of order n. Denote the Latin squares by $A_1,A_2,\dots,A_k$. The transformation creates $k$ blocks, say $B_1,B_2,\dots,B_k$ of permutations on n symbols.
Moreover, the pairwise Hamming distance between blocks $B_i,B_j$ for all $i,j,~(1 \le i,j, \le k, ~i\neq j)$ is $n-1$. 
We repeat this transformation for $k$ MOLS of order $m$ to create the block decomposition $E_1,E_2,\dots, E_k$ of permutations on $Z_m$, with pairwise Hamming distance $m-1$.
By Lemma \ref{l:kron3}, $M(nm+1,nm)\geq knm$.
\end{proof}

\bigskip
\noindent Example \ref{example_kron} shows several new bounds obtained by Corollary \ref{cor:kron1}. 
Additional new results obtained by Corollaries \ref{cor:kron1} and \ref{cor:kron2} are listed in Tables \ref{tb:n_1a} and \ref{tb:n_1b}.

\begin{changemargin}{.5cm}{.5cm}
\begin{example}
\label{example_kron} 
A sample of results from Corollary \ref{cor:kron1} with $A=AGL(1,p)$ and $B=AGL(1,q)$.
\begin{enumerate} [\indent (a)]
\item $M(117,116)\ge 8\cdot 117=936$ by using $p=9$ and $q=13$. So $M(118,117) \ge 936 $.
\item $M(171,170)\ge 8\cdot 171=1368$ by using $p=9$ and $q=19$. So $M(172,171) \ge 1,368 $.
\item $M(187,186)\ge 10\cdot 187=1870$ by using $p=11$ and $q=17$. So $M(188,187) \ge 1,871 $.
\item $M(299,298)\ge 12\cdot 299=3588$ by using $p=13$ and $q=23$. So $M(300,299) \ge 3,588 $.
\item $M(575,574)\ge 22\cdot 575=12650$ by using $p=23$ and $q=25$. So $M(576,575) \ge 12,650 $.
\end{enumerate}
\end{example}
\end{changemargin}

\bigskip

\section{Algorithms for Selecting Partitions}

\label{s:algorithms}

In Sections \ref{s:seq_pe}, \ref{s:par_pe} and \ref{s:kron}, we described three new enhancements of the partition and extension operation which are used for transforming a distance-$d$ partition system $\Pi=(\M,\P,\Q)$ on $Z_n$, for some positive integer $d$, into a new PA on $Z_{n+r}$ for positive integers $r$, such that the Hamming distance of the new PA is at least $d'$ for some $d' \ge d$.  
The size of a PA resulting from the application of any of these techniques to a particular distance-$d$ partition  system, $\Pi=(\M,\P,\Q)$, is of course entirely dependent on the choice of $\M$, $\P$, and $\Q$. 
Exhaustive search for high yield partitions $\P$ and $\Q$ amounts to trying all possible partitions of $Z_n$.
Similarly, selecting a productive set of PAs to include in $\M$ involves selecting sets from partitions of $S_n$, the symmetric group of permutations on $n$ symbols. 
Clearly, any sort of exhaustive search is infeasible.

This leads to a natural question: how to select the sets  $\M$, $\P$, and $\Q$. 
We now describe several techniques we have found useful for selecting  partitions for the set $\P$ (or, equivalently, $\Q$), and finding  PAs for the set $\M$.

In Sections \ref{ss:greedy_partition_sel} and \ref{ss:ilp_partition_sel}, we turn our attention to methods for finding partitions of $Z_n$. 
Such partitions can be fruitful candidates for either for $\P$ or $\Q$. 
We describe two approaches. 
Both approaches start with a given partition of symbols $\Q$ and a given collection of PAs $\M=(M_1, M_2, \dots, M_{k+1})$ on $Z_n$, for some positive integer $k$, that satisfies Property \ref{dist-d prop1} of the definition of a distance-$d$ partition system. 
Section \ref{ss:greedy_partition_sel} describes a greedy algorithm that uses a fixed partition of symbols $\Q$ and greedily creates a partition of positions, $\P$. 
Section \ref{ss:ilp_partition_sel} describes an optimization approach that uses Integer Linear Programming to find a fruitful partition of positions, $\P$. 
To describe the techniques, we focus on creating a partition of positions $\P$, however, the same techniques can be used for creating a partition of symbols $\Q$ instead.
We have experimented with both methods and have obtained new lower bounds for $M(n,d)$ which are included in Section \ref{s:results}. 

Section \ref{ss:coset_search} describes methods we have used for searching for fruitful PAs to include in $\M$. 
New lower bounds obtained by this method are included in Section \ref{s:results}.

\subsection{A Greedy Approach to Partition Selection}
\label{ss:greedy_partition_sel}

We have developed a greedy algorithm for finding a partition of positions $\P$, which approaches an intractable search problem by fixing both the partition of symbols, $\mathcal{Q}$, and the collection of PAs, $\M$, then greedily creating $\P$, a  partition of positions. In this way, the  search space is restricted, at the cost of possibly missing an optimum solution.

Our algorithm creates a partition positions $\P$, of $Z_n$, that maximizes $covered(M_i)$ for all $i$. 
The input for the algorithm is a fixed partition of symbols $\Q$ of $Z_n$, and a collection of PAs on $Z_n$, $\mathcal{M}=(M_1,M_2,\dots,M_k)$, that satisfies properties I and II of a distance-$d$ partition system for some $d<n$.  
We fix $\mathcal{Q}=(Q_1,Q_2,\dots,Q_k)$ for some $k\le\sqrt{n}$ where 
$Q_1=\{0,1,\dots, k-1\},Q_2=\{k,\dots, 2k-1\}, \dots\ Q_k=\{k^2-k,\dots, k^2-1\}$.

The algorithm starts with a set of subsets of positions $\{P_1,P_2,\dots,P_k\}$ where $P_i=\emptyset$ for all $i~ (0\le i\le k-1)$.
The algorithm then iterates to find a partition of positions $\P$ that represents a local maximum for the number of covered permutations.
At each iteration, an unused position, $r$, is selected. 
Let $M'_i=M_i \setminus covered(M_i)$. That is, $M'_i$ is the set of permutations $\{\sigma\}$ in $M_i$ for which there is no position $p \in P_i$ such that $\sigma(p)=q$ for some $q \in Q_i$. 
For each $i~ (1\le i\le k)$, we count the number of covered permutations for $(M'_i,P_i\cup\{r\},Q_i)$. If the number of covered permutations is maximized for some $i=i^*$, then we add $r$ to $P_{i^*}$. 
The algorithm stops when there are no more unused positions.

The resulting partition $\P$, together with $\Q$ and $\M$ form a distance-$d$ partition system for $Z_n$, $\Pi=(\M,\P,\Q)$.
So, by Theorem \ref{th:simple_pe}, $hd(ext(\Pi))\ge d$. 
There are several instances for which our greedy approach results in a partition system $\Pi$ that provides full coverage, that is, for all $i~ (1\le i\le k),~ covered(M_i)=M_i$.
When $\Pi$ is derived from large PAs such as $AGL(1,q)$, for $q$, a power of a prime, improved lower bounds can be achieved for $M(q+1,d)$.
A list of results is included in Tables  
\ref{tb:n_d_all}, \ref{tb:n_1a} and \ref{tb:n_1b}.

\subsection{An Optimization Approach to Partition Selection} \label{ss:ilp_partition_sel}

We describe another approach for finding a partition of positions $\P$, which casts the search for $\P$  as an optimization problem. 
Like the greedy method, our optimization approach starts with a given partition $\mathcal{Q}$ of symbols, and a collection $\M$ of PAs that satisfies properties I and II of a distance-$d$ partition system for some $d<n$. 
We encode the search for $\mathcal{P}$ as an Integer Linear Program (ILP) and use an off-the-shelf \textit{solver} to explore the entire search space of partitions for $\mathcal{P}$.
There are several commercial solvers \cite{cplex2009v12,gurobi} capable of solving large ILP problems efficiently.
We have chosen the Gurobi optimizer \cite{gurobi} for our computations. 

We now describe our ILP encoding. The input is a partition of symbols $\mathcal{Q}$ and a collection $\mathcal{M}$ of blocks (PAs) on $n$ symbols. 
Let $k$ be the number of blocks.
Let $c_{i,j}$ be a binary variable indicating that permutation $j$ of block $i$ is covered. 
Let $u(i)$ be a function that maps the block index $i$ to the number of permutations in it. 
Let $b_{i,p}$ be a binary variable indicating that position $p$ is assigned to block $i$.

\begin{figure}[hbt]

\begin{align}
\vspace*{-\baselineskip}
    \underset{c_{i,j}}{\text{maximize~}}  & \sum_{i=0}^{k-1} \sum_{j=0}^{u(i)-1} {c_{i,j}} \label{eq:obj_pe}\\
   \text{ subject to} \notag \\
   & \sum_{i=0}^{k-1}b_{i,p}=1;~\forall p; \label{eq:c1_pe}\\
   & \sum_{y\in Q_i} \mathds{1}_{\sigma_{p,y}} \cdot b_{i,p} \geq {c_{i,j}}; ~\forall i,j,p; \text { and} \label{eq:c2_pe}\\
   & \sum_{i=0}^{k-1}\sum_{p=0}^{n-1} b_{i,p}=n; \label{eq:c3_pe}\\
	\text{where~~~~} &
\scriptsize
\mathds{1}_{\sigma_{p,y}}= 
\begin{cases}
   1  & \text{if } \sigma[p]=y\\
   0 & \text{otherwise.}
\end{cases} 
\label{eq:cov_per}
\end{align}

\caption{An Integer Linear Program for selecting partitions} 
\label{ILPfig}
\end{figure}

Equation (\ref{eq:obj_pe}) is the objective function to be maximized, that is, the total number of covered permutations in all blocks in $\M$. The optimization is subject to three constraints: 
\begin{itemize}  
\item Constraint (\ref{eq:c1_pe}) assures that the resulting partition $\P$ assigns a position to exactly one block.
\item Constraint (\ref{eq:c2_pe}) establishes that permutation $j$ in block $i$ is covered when at least one of its symbols listed in $Q_i$ appears in position $p$, and $p$ is assigned to this block $i$. 
\item Constraint (\ref{eq:c3_pe}) assures that every position has been assigned to some block. 
\end{itemize}
Constraints (\ref{eq:c1_pe}) and (\ref{eq:c3_pe}) effectively ensure that the solution is a partition. 
Equation (\ref{eq:cov_per}) defines an indicator function that states whether or not a permutation $\sigma$ is covered by checking if symbol $y$ appears at position $p$.

Our Integer Linear Program has provided many new lower bounds for $M(n,d)$, and has has outperformed our greedy approach in several instances. See Tables \ref{tb:n_d_all}, \ref{tb:n_1a} and \ref{tb:n_1b}.

\subsection{Methods for Coset Search} \label{ss:coset_search}
We have used several methods for coset search, including the {\em coset method} \cite{bereg2015constructing} and Integer Linear Programming. 

Given a group $G$ on $Z_n$ for some $n$, the coset method creates a collection of PAs $\M$ to be used for partition and extension by randomly searching for cosets of $G$ at a specified pairwise Hamming distance $d$.  
The group $G=M_1$, with its cosets, $M_2,M_3,...,$ comprise $\M=(M_1,M_2,M_3,...,)$ in a distance-$d$ partition system $\Pi$. 
When the starting group $G$ is large, the coset method often produces a productive collection of PAs for $\M$.


Table \ref{tb:n_d_coset} shows the lower bounds obtained by applying Theorem \ref{th:simple_pe} to new permutation arrays computed using the coset method. 
For example, for our new lower bound for $M(43,37)$, we start with the projective general linear group $G = PGL(2,41)$, which has 68,880 permutations on $Z_{42}$, and looked for cosets of $G$ at Hamming distance 36. 
We were able to find five cosets, $M_2,M_3,M_4,M_5,M_6,$ which together with the group $G=M_1$ gives a collection of 6 blocks with 68,800 permutations each, giving a total of  413,280 permutations at Hamming distance 36. 
This gives $\M=(M_1, M_2, \dots, M_6)$. 
We were also able to find a partition of positions $\P$ and a partition of symbols $\Q$, which, together with $\M$ forms a distance-37 partition system $\Pi=(\M,\P,\Q)$ for $Z_{42}$. 
Using simple partition and extension on $\Pi$, we obtained 369,948 permutations on 43 symbols with Hamming distance 37.
That is, we show that $M(43,37) \ge 369,948$, which is an improvement over the previous lower bound of 176,988.

\begin{table}[ht]

\centering
\vspace*{4mm}
\begin{tabular}{|r r r r |}
\hline
$n$ & $d$ & \bfseries PREV & \bfseries NEW \\
\hline
\hline
43  & 37   & 176,988  & 369,948\\
49  & 43   & 207,552  & 415,062 \\
51  & 44   & 235,200  & 687,903 \\
51  & 45   & 235,200  & 470,347 \\
61  & 54   & 410,640  & 1,181,794\\
69  & 62   & 601,392  & 1,500,426\\
\hline
\end{tabular}
\caption{New $M(n,d)$ lower bounds obtained by applying Theorem \ref{th:simple_pe} to PAs generated by the coset method \cite{bereg2015constructing}. Column \textbf{PREV} shows previously known bounds (obtained from rudimentary parallel partition and extension, by applying Theorem \ref{thm:rudi}).
Column \textbf{NEW} shows new bounds obtained through Theorem \ref{th:simple_pe}. }
\label{tb:n_d_coset}
\end{table}

We have also searched for fruitful PAs by formulating the coset search problem as a constraint satisfaction problem, implemented as an Integer Linear Program. 
Given a group $G$ on $Z_n$, where $hd(G)\ge d$, let $d'$ be the target Hamming distance between a coset representative $\pi \in S_n$ and the group $G$. 
Let $X=Z_n \times Z_n=\{(0,0),(0,1),\dots,(i,j),\dots,(n-1,n-1)\}$. 
The set $X$ represents all possible pairs of positions and symbols assignable to the coset representative $\pi$.

Create a binary variable $x_{i,j}$ for each element in the set $X$ indicating that if the variable $x_{i,j}$ is true, then $\pi(i)=j$. 
The Integer Linear Program is:

\begin{align}
   & \underset{x_{i,j}}{\text{maximize}} ~\sum_{i=0}^{n-1} \sum_{j=0}^{n-1} {x_{i,j}} \label{eq:obj_search}\\
   &\text{subject to} \notag \\
   & \sum_{j=0}^{n-1}x_{i,j}=1;~\forall i\in Z_n, \label{eq:c1_search}\\
   & \sum_{i=0}^{n-1}x_{i,j}=1;~\forall j\in Z_n, \text{ and} \label{eq:c2_search}\\
   & \sum_{i=0}^{n-1}\sum_{j=0}^{n-1} ~\mathds{1}_{\sigma_{i,j}} \cdot x_{i,j} \leq n-d;~\forall \sigma \in G, \label{eq:c3_search}
\end{align}

\begin{align}
\text{where~~~~} &
\mathds{1}_{\sigma_{i,j}}= 
\begin{cases}
   1  & \text{if } \sigma(i)=j\\
   0 & \text{otherwise}
\end{cases}	
\end{align}

The objective function (\ref{eq:obj_search}) is designed to make the ILP solver assign as many binary variables $x_{i,j}$ true as possible. 
This objective function alone would produce a solution that is not a permutation. 
For this reason constraints (\ref{eq:c1_search}) and (\ref{eq:c2_search}) ensure that exactly one symbol $j$ is assigned to every position $i$ and that every symbol $j$ is assigned to exactly one position $i$, respectively, so the solution is indeed a permutation on $Z_n$. 
Constraint (\ref{eq:c3_search}) requires the solution to be at Hamming distance at least $d'$ from every permutation in $G$. 
This is encoded by limiting the number of agreements, $n-d'$, between a candidate solution and each of the permutations in $G$.

Table \ref{tb:coset} gives a detailed view of new lower bounds for $M(n,d)$, resulting from our coset search techniques. 
For each new result, the group, $G$ and the number of cosets is shown. 
The subscript $j$ in the column labeled \textbf{NEW} indicates that the cosets were found by the Integer Linear Program described in section \ref{ss:coset_search} \cite{luis17}.
The subscript $c$ indicates that the cosets were found by the coset method \cite{bereg2015constructing}.

\begin{table} [H]


\centering
\vspace*{4mm}
\begin{tabular}{|c c c c r r |}
\hline
\bfseries $n$ & \bfseries $d$ & {\em Group} & {\em Num Cosets} & \bfseries PREV & \bfseries NEW \\
\hline
\hline
$18$	&	$13$	&	$PGL(2,17)$	&	$6$	&		$24,480$	&	$29,376_j$	\\
$24$	&	$19$	&	$PGL(2,23)$	&	$3$	&		$24,288$	&	$36,432_j$\\	
$26$	&	$20$	&	$PGL(2,25)$	&	$15$	&		$202,800$	&	$234,000_j$	\\
$26$	&	$21$	&	$PGL(2,25)$	&	$3$	&		$31,200$	&	$46,800_j$	\\
$28$	&	$22$	&	$PGL(2,27)$	&	$14$	&		$235,872$	&	$275,184_j$	\\
$30$	&	$24$	&	$PGL(2,29)$	&	$12$	&		$170,520$	&	$292,320_j$	\\
$32$	&	$25$	&	$PGL(2,31)$	&	$44$	&		$372,992$	&	$1,309,440_j$	\\
$33$	&	$27$	&	$P\Gamma L(2,32)$	&	$2$	&		$97,440$	&	$327,360_j$\\	
34 & 27 & $P\Gamma L(2,32)$	&	$15$	&	$2,127,840$	&	$2,455,200_c$\\	
$38$	&	$32$	&	$PGL(2,37)$	&	$6$	&		$202,464$	&	$303,696_j$	\\
$38$	&	$30$	&	$PGL(2,37)$	&	$129$	&	$1,265,400$	&	$6,529,464_c$	\\ 
$42$	&	$34$	&	$PGL(2,41)$	&	$73$	&	$888,729$	&	$5,028,240_c$	\\
$42$	&	$35$	&	$PGL(2,41)$	&	$28$	&		$206,640$	&	$1,928,640_j$	\\
$42$	&	$36$	&	$PGL(2,41)$	&	$6$	&		$206,640$	&	$413,280_j$	\\
$44$	&	$37$	&	$PGL(2,43)$	&	$25$	&		$413,280$	&	$1,986,600_j$	\\
$48$	&	$42$	&	$PGL(2,47)$	&	$4$	&		$207,552$	&	$415,104_j$	\\
$49$	&	$42$	&	$PGL(2,47)$	&	$14$	&	$207,552$	&	$1,452,864_c$	\\
$50$	&	$42$	&	$PGL(2,49)$	&	$43$	&	$207,552$	&	$5,056,800_c$	\\
$50$	&	$43$	&	$PGL(2,49)$	&	$18$	&		$207,552$	&	$2,116,800_j$	\\
$50$	&	$44$	&	$PGL(2,49)$	&	$4$	&		$103,776$	&	$470,400_j$	\\
$54$	&	$47$	&	$PGL(2,53)$	&	16	&		$1,339,416$	&	$2,381,184_j$	\\
$54$	&	$48$	&	$PGL(2,53)$	&	$3$	&		$297,648$	&	$446,472_j$	\\
$55$	&	$48$	&	$PGL(2,53)$	&	$10$	&	$297,648$	&	$1,488,240_c$	\\
$55$	&	$49$	&	$PGL(2,53)$	&	$3$	&		$297,648$	&	$446,472_j$	\\
$62$	&	$54$	&	$PGL(2,61)$	&	$38$	&	$821,280$	&	$8,622,960_c$	\\
$62$	&	$55$	&	$PGL(2,61)$	&	$6$	&	$821,280$	&	$1,361,520_c$	\\
$68$	&	$60$	&	$PGL(2,67)$	&	$29$	&	$821,280$	&	$8,720,184_c$	\\
$68$	&	$61$	&	$PGL(2,67)$	&	$5$	&	$524,160$	&	$1,503,480_c$	\\
$68$	&	$62$	&	$PGL(2,67)$	&	$2$	&	$524,160$	&	$601,392_j$	\\
$72$	&	$64$	&	$PGL(2,71)$	&	$17$	&	$888,729$	&	$6,083,280_c$	\\
$72$	&	$65$	&	$PGL(2,71)$	&	$4$	&	$357,840$	&	$1,431,360_c$	\\

\hline
\end{tabular}
\caption{New lower bounds for $M(n,d)$ using PAs generated by the coset method \cite{bereg2015constructing} and by ILP approximation described in Section \ref{ss:coset_search}. 
Columns: 
{\em Group} denotes starting group, 
{\em Num Cosets} denotes the number of cosets, 
\textbf{PREV} denotes the previously known bound, and 
\textbf{NEW} denotes the new bound.
Subscript legend: $c$-coset method (random coset search) \cite{bereg2015constructing};  $j$- ILP coset search (see Section \ref{ss:coset_search}).
}
\label{tb:coset}
\end{table}

\section{Summary of New Results } 
\label{s:results}

We have computed many new lower bounds for $M(n,d)$ for various $n$ and $d$ using our new techniques for partition and extension, namely:
sequential partition and extension (Corollary \ref{cor:iterative_pe_general} and Theorem \ref{th:sequential_pe}), 
parallel partition and extension (Theorem \ref{thm:rudi}, \ref{thm:parallel_pe_2}), 
and modified Kronecker product (Corollaries \ref{cor:kron1}, and \ref{cor:kron2}).
These techniques are described in Sections \ref{s:seq_pe}, \ref{s:par_pe}, and \ref{s:kron}. 
We have also used our earlier technique of simple partition and extension (see Theorem \ref{th:simple_pe} \cite{bms17}) to generate new lower bounds. 
The use of partition and extension requires, as input, a partition of positions and a separate partition of symbols. 
We have used our greedy and ILP algorithms, (described in Section \ref{ss:greedy_partition_sel} and \ref{ss:ilp_partition_sel}), to obtain fruitful partitions of positions for many $n$. 
We have described methods for generating good collections of PAs for our partition and extension techniques. (See Section \ref{ss:coset_search}).
 
We summarize all of our new lower bounds for $M(n,d)$, for $d<n-1$, in Table \ref{tb:n_d_all} for the sake of easy referencing. 
We also report experimental results and provide new tables of lower bounds for $M(n,n-1)$, for many integers $n<600$. 
Due to the large number of results, we show these separately from our results for $M(n,d)$, for $d<n-1$.  
Tables \ref{tb:n_1a} and \ref{tb:n_1b} show new lower bounds for $M(n,n-1)$ computed by our partition and extension techniques.
Columns \textbf{PREV} and \textbf{NEW} in Tables \ref{tb:n_1a} and \ref{tb:n_1b} denote the previous and the new bound, respectively.
The previous lower bounds are either from an earlier use of simple partition and extension \cite{bms17}, and are denoted with a subscript $P$, or are derived from known numbers of mutually orthogonal squares (MOLS) \cite{colbourn2006handbook}, and are denoted with a subscript $M$. 
It should be noted that there are other known lower bounds for $M(n,n-1)$, for integers $n$ not listed in Tables \ref{tb:n_1a} and \ref{tb:n_1b}. 
They have been previously reported in \cite{bms17, colbourn2006handbook}, and \cite{jani15}.
The subscripts in the \textbf{NEW} column indicate the method for generating either the partition of positions $\P$ or the collection of PAs $\M$.
Subscript $g$ indicates that $\P$ was computed using the greedy partition selection algorithm. (See Section \ref{ss:greedy_partition_sel}).
Subscript $i$ indicates that $\P$ was computed using the Integer Linear Program for partition selection. (See Section \ref{ss:ilp_partition_sel}).
Subscript $a$ indicates new bounds described in \cite{manuscript}. 
Subscript $k$ indicates the collection of PAs $\M$ is obtained by modified Kronecker product. (See Section \ref{s:kron}).

In conclusion, we offer the following conjecture about the relationship between $N(n)$, the known lower bound on the number of MOLS of side $n$ and $M(n,n-1)$:


\begin{equation}
\label{eq:conjecture_equation}
\textbf{Conjecture:   } M(n,n-1)\ge  (n-1)\cdot \min(\lfloor\sqrt{n-1}\rfloor,N(n-1)).
\end{equation}

\noindent This conjecture is based on our computational results. We verified that the conjecture is true for all $n \le 600$, except the four cases listed in Table \ref{tb:conjecture}. Although these may seem to be counterexamples for the conjecture, we believe the computed values can be improved, and therefore, the conjecture validated for all $n \le 600$.

\begin{table}[ht]

\centering
\vspace*{4mm}
\begin{tabular}{|c c c c |}
\hline
$n$ & $d$ & \bfseries Computed & \bfseries Conjectured \\
\hline
\hline
145  & 144   & 1,429  & 1,440 \\
177  & 176   & 2,214  & 2,288 \\
225  & 224   & 2,902  & 2,912 \\
254  & 253   & 3,027  & 3,036 \\
\hline
\end{tabular}
\caption{A comparison of experimentally computed $M(n,n-1)$ lower bounds to  conjectured lower bounds for four cases that (so far) do not agree with the conjecture. 
Column \textbf{Computed} shows known bounds obtained from techniques described in this paper.
Column \textbf{Conjectured} shows conjectured bounds from Equation \ref{eq:conjecture_equation}. }
\label{tb:conjecture}
\end{table}

\def\c{\ref{tb:coset}}


\begin{table}[htb]
\centering
\vspace*{2mm}
\begin{tabular}{| r r r r | r r r r |r r r r |}
\hline
$n$ & $d$ & \bfseries PREV & \bfseries NEW & $n$ & $d$ & \bfseries PREV & \bfseries NEW & $n$ & $d$ & \bfseries PREV & \bfseries NEW\\
\hline\hline
18 & 13 & 24,480 & 29,376$_{\ref{tb:coset}}$ & 53 & 47 & 148,824 & 470,400$_{\ref{tb:new_parallel_pe}}$ & 171 & 169 & 2,354 & 27,330$_{\ref{tb:n_2}}$ \\
24 & 19 & 24,288 & 36,432$_{\ref{tb:coset}}$ & 54 & 46 & 8,036,496 & 8,334,144$_{\ref{tb:coset}}$ & 175 & 173 & 2,354 & 19,792$_{\ref{tb:n_2}}$ \\
26 & 20 & 202,800 & 234,000$_{\ref{tb:coset}}$ & 54 & 47 & 1,339,416 & 2,381,184$_{\ref{tb:coset}}$ & 183 & 181 & 2,533 & 21,994$_{\ref{tb:n_2}}$ \\
26 & 21 & 31,200 & 46,800$_{\ref{tb:coset}}$ & 54 & 48 & 297,648 & 446,472$_{\ref{tb:coset}}$ & 195 & 193 & 2,758 & 25,022$_{\ref{tb:n_2}}$ \\
28 & 22 & 235,872 & 275,184$_{\ref{tb:coset}}$ & 55 & 48 & 297,648 & 1,488,240$_{\ref{tb:coset}}$ & 201 & 199 & 2,867 & 25,427$_{\ref{tb:n_2}}$ \\
30 & 24 & 170,520 & 292,320$_{\ref{tb:coset}}$ & 55 & 49 & 297,648 & 446,472$_{\ref{tb:coset}}$ & 213 & 211 & 3,170 & 30,288$_{\ref{tb:n_2}}$ \\
30 & 26 & 24,360 & 58,968$_{\ref{tb:new_parallel_pe}}$ & 55 & 53 & 423 & 2,461$_{\ref{tb:n_2}}$ & 225 & 223 & 3,421 & 32,728$_{\ref{tb:n_2}}$ \\
32 & 25 & 372,992 & 1,309,440$_{\ref{tb:coset}}$ & 56 & 50 & 205,320 & 446,472$_{\ref{tb:new_parallel_pe}}$ & 231 & 229 & 3,548 & 33,779$_{\ref{tb:n_2}}$ \\
33 & 27 & 97,440 & 327,360$_{\ref{tb:coset}}$ & 61 & 54 & 410,640 & 1,181,794$_{\ref{tb:n_d_coset}}$ & 235 & 233 & 3,625 & 35,001$_{\ref{tb:n_2}}$ \\
34 & 27 & 2,127,840 & 2,455,200$_{\ref{tb:coset}}$ & 62 & 54 & 821,280 & 8,622,960$_{\ref{tb:coset}}$ & 245 & 243 & 3,475 & 43,717$_{\ref{tb:n_2}}$ \\
34 & 32 & 192 & 945$_{\ref{tb:n_2}}$ & 62 & 55 & 821,280 & 1,361,520$_{\ref{tb:coset}}$ & 253 & 251 & 4,075 & 40,094$_{\ref{tb:n_2}}$ \\
38 & 30 & 1,265,400 & 6,529,464$_{\ref{tb:coset}}$ & 63 & 61 & 1,514 & 3,306$_{\ref{tb:n_2}}$ & 259 & 257 & 4,222 & 43,268$_{\ref{tb:n_2}}$ \\
38 & 32 & 202,464 & 303,696$_{\ref{tb:coset}}$ & 66 & 64 & 576 & 4,029$_{\ref{tb:n_2}}$ & 265 & 263 & 4,342 & 44,733$_{\ref{tb:n_2}}$ \\
39 & 37 & 255 & 1,301$_{\ref{tb:n_2}}$ & 68 & 60 & 821,280 & 8,720,184$_{\ref{tb:coset}}$ & 273 & 271 & 4,548 & 46,268$_{\ref{tb:n_2}}$ \\
40 & 34 & 68,880 & 287,437$_{\ref{tb:new_parallel_pe}}$ & 68 & 61 & 524,160 & 1,503,480$_{\ref{tb:coset}}$ & 279 & 277 & 4,701 & 49,243$_{\ref{tb:n_2}}$ \\
42 & 34 & 888,729 & 5,028,240$_{\ref{tb:coset}}$ & 68 & 62 & 524,160 & 601,392$_{\ref{tb:coset}}$ & 285 & 283 & 4,868 & 51,571$_{\ref{tb:n_2}}$ \\
42 & 35 & 206,640 & 1,928,640$_{\ref{tb:coset}}$ & 69 & 62 & 601,392 & 1,500,426$_{\ref{tb:n_d_coset}}$ & 291 & 289 & 5,202 & 80,385$_{\ref{tb:n_2}}$ \\
42 & 36 & 206,640 & 413,280$_{\ref{tb:coset}}$ & 69 & 67 & 594 & 3,965$_{\ref{tb:n_2}}$ & 295 & 293 & 5,088 & 54,572$_{\ref{tb:n_2}}$ \\
43 & 37 & 176,988 & 369,948$_{\ref{tb:n_d_coset}}$ & 70 & 63 & 524,160 & 1,503,462$_{\ref{tb:new_parallel_pe}}$ & 309 & 307 & 5,539 & 60,715$_{\ref{tb:n_2}}$ \\
44 & 37 & 413,280 & 1,986,600$_{\ref{tb:coset}}$ & 72 & 64 & 888,729 & 6,083,280$_{\ref{tb:coset}}$ & 315 & 313 & 5,634 & 60,952$_{\ref{tb:n_2}}$ \\
44 & 38 & 68,880 & 397,198$_{\ref{tb:new_parallel_pe}}$ & 72 & 65 & 357,840 & 1,431,360$_{\ref{tb:coset}}$ & 319 & 317 & 5,793 & 67,379$_{\ref{tb:n_2}}$ \\
45 & 39 & 103,776 & 413,280$_{\ref{tb:new_parallel_pe}}$ & 75 & 73 & 667 & 4,747$_{\ref{tb:n_2}}$ & 333 & 331 & 6,091 & 70,696$_{\ref{tb:n_2}}$ \\
45 & 43 & 270 & 1,726$_{\ref{tb:n_2}}$ & 85 & 83 & 812 & 6,116$_{\ref{tb:n_2}}$ & 339 & 337 & 6,280 & 69,485$_{\ref{tb:n_2}}$ \\
46 & 39 & 103,776 & 551,040$_{\ref{tb:new_parallel_pe}}$ & 91 & 89 & 902 & 6,709$_{\ref{tb:n_2}}$ & 345 & 343 & 5,205 & 89,272$_{\ref{tb:n_2}}$ \\
48 & 42 & 207,552 & 415,104$_{\ref{tb:coset}}$ & 99 & 97 & 1,017 & 8,206$_{\ref{tb:n_2}}$ & 351 & 349 & 6,642 & 76,195$_{\ref{tb:n_2}}$ \\
49 & 42 & 207,552 & 1,452,864$_{\ref{tb:coset}}$ & 105 & 103 & 1,119 & 9,239$_{\ref{tb:n_2}}$ & 355 & 353 & 6,746 & 77,215$_{\ref{tb:n_2}}$ \\
49 & 43 & 207,552 & 415,062$_{\ref{tb:n_d_coset}}$ & 111 & 109 & 1,187 & 9,990$_{\ref{tb:n_2}}$ & 363 & 361 & 7,220 & 125,709$_{\ref{tb:n_2}}$ \\
50 & 42 & 207,552 & 5,056,800$_{\ref{tb:coset}}$ & 115 & 113 & 1,277 & 11,142$_{\ref{tb:n_2}}$ & 369 & 367 & 7,108 & 83,418$_{\ref{tb:n_2}}$ \\
50 & 43 & 207,552 & 2,116,800$_{\ref{tb:coset}}$ & 123 & 121 & 1,452 & 13,996$_{\ref{tb:n_2}}$ & 375 & 373 & 7,298 & 87,434$_{\ref{tb:n_2}}$ \\
50 & 44 & 103,776 & 470,400$_{\ref{tb:coset}}$ & 133 & 131 & 1,554 & 11,604$_{\ref{tb:n_2}}$ & 385 & 383 & 7,428 & 90,213$_{\ref{tb:n_2}}$ \\
51 & 44 & 235,200 & 687,903$_{\ref{tb:n_d_coset}}$ & 141 & 139 & 1,723 & 13,522$_{\ref{tb:n_2}}$ & 391 & 389 & 7,690 & 90,991$_{\ref{tb:n_2}}$ \\
51 & 45 & 235,200 & 470,347$_{\ref{tb:n_d_coset}}$ & 153 & 151 & 1,923 & 16,118$_{\ref{tb:n_2}}$ & 411 & 409 & 8,240 & 104,098$_{\ref{tb:n_2}}$ \\
51 & 49 & 392 & 2,308$_{\ref{tb:n_2}}$ & 159 & 157 & 2,051 & 16,666$_{\ref{tb:n_2}}$ & 514 & 512 & 11,264 & 197,859$_{\ref{tb:n_2}}$ \\
52 & 46 & 148,824 & 470,397$_{\ref{tb:new_parallel_pe}}$ & 165 & 163 & 2,185 & 17,632$_{\ref{tb:n_2}}$ & 531 & 529 & 12,696 & 271,043$_{\ref{tb:n_2}}$ \\

\hline
\end{tabular}
\caption{An aggregated table showing our new lower bounds for $M(n,d)$, for  $n<550$ and $d<n-1$. The subscripts give the tables containing more details about the new results.}
\label{tb:n_d_all}
\end{table}


\newpage
\begin{table}[htb]

\centering
\vspace*{2mm}
\begin{tabular}{|r r r | r r r | r r r|}
\hline
$n$ & \bf Prev & \bf New & 
$n$ & \bf Prev & \bf New & 
$n$ & \bf Prev & \bf New \\
\hline\hline
26 & 133$_P$ & 150$_a$&132 & 1508$_P$ & 1572$_g$&212 & 3026$_P$ & 3172$_i$\\
28 & 140$_M$ & 144$_i$&134 & 804$_M$ & 931$_g$&214 & 1284$_M$ & 1491$_g$\\
30 & 170$_P$ & 173$_g$&138 & 1614$_P$ & 1696$_g$&218 & 1308$_M$ & 1736$_g$\\
33 & 183$_P$ & 192$_a$&140 & 1640$_P$ & 1726$_i$&220 & 1320$_M$ & 2190$_g$\\
34 & 136$_M$ & 165$_g$&142 & 852$_M$ & 987$_g$&222 & 1332$_M$ & 2652$_g$\\
38 & 254$_P$ & 255$_g$&145 & 1015$_M$ & 1429$_i$&224 & 3260$_P$ & 3475$_i$\\
42 & 282$_P$ & 286$_g$&146 & 876$_M$ & 1015$_g$&225 & 1800$_M$ & 2902$_i$\\
44 & 296$_P$ & 307$_g$&148 & 888$_M$ & 1029$_g$&226 & 1356$_M$ & 1800$_k$\\
46 & 184$_M$ & 270$_g$&150 & 1818$_P$ & 1905$_g$&228 & 3380$_P$ & 3482$_i$\\
50 & 300$_M$ & 392$_a$&152 & 1832$_P$ & 1946$_g$&230 & 3512$_P$ & 3567$_g$\\
51 & 255$_M$ & 300$_g$&155 & 1085$_M$ & 1232$_g$&234 & 3602$_P$ & 3673$_i$\\
54 & 408$_P$ & 423$_g$&156 & 936$_M$ & 1085$_g$&236 & 1416$_M$ & 1645$_g$\\
58 & 361$_P$ & 399$_i$&158 & 1922$_P$ & 2052$_g$&238 & 1428$_M$ & 1659$_g$\\
60 & 481$_P$ & 493$_g$&159 & 954$_M$ & 1106$_g$&240 & 3656$_P$ & 3803$_i$\\
62 & 478$_P$ & 519$_g$&161 & 1377$_P$ & 1440$_i$&242 & 3716$_P$ & 3864$_g$\\
65 & 455$_M$ & 576$_a$&162 & 972$_M$ & 1127$_g$&244 & 1464$_M$ & 3483$_a$\\
66 & 380$_P$ & 455$_g$&164 & 2042$_P$ & 2185$_g$&246 & 1476$_M$ & 1715$_g$\\
68 & 568$_P$ & 594$_g$&166 & 1153$_P$ & 1155$_g$&248 & 1736$_M$ & 2964$_g$\\
72 & 588$_P$ & 637$_g$&168 & 2070$_P$ & 2267$_g$&250 & 1500$_M$ & 1743$_g$\\
74 & 620$_P$ & 667$_g$&170 & 1020$_M$ & 2366$_a$&252 & 3932$_P$ & 4075$_g$\\
76 & 456$_M$ & 525$_g$&172 & 1032$_M$ & 1368$_k$&254 & 2286$_M$ & 3027$_i$\\
80 & 720$_M$ & 755$_g$&174 & 2316$_P$ & 2358$_i$&255 & 1785$_M$ & 2286$_g$\\
82 & 656$_M$ & 810$_a$&177 & 1593$_M$ & 2214$_i$&258 & 4066$_M$ & 4222$_g$\\
84 & 776$_P$ & 812$_g$&178 & 1068$_P$ & 1593$_g$&260 & 1560$_M$ & 3108$_g$\\
90 & 866$_P$ & 902$_g$&180 & 2404$_P$ & 2500$_g$&264 & 4228$_P$ & 4351$_i$\\
92 & 552$_M$ & 637$_g$&182 & 1092$_P$ & 2533$_g$&266 & 1862$_M$ & 2120$_g$\\
98 & 956$_P$ & 1017$_g$&186 & 1619$_P$ & 1665$_g$&268 & 1876$_M$ & 2670$_g$\\
102 & 1030$_P$ & 1101$_g$&188 & 1128$_M$ & 1870$_k$&270 & 4318$_M$ & 4521$_i$\\
104 & 1070$_P$ & 1119$_g$&190 & 1140$_M$ & 1512$_g$&272 & 4408$_M$ & 4575$_i$\\
106 & 636$_M$ & 735$_g$&192 & 2638$_P$ & 2767$_i$&274 & 1644$_M$ & 3873$_i$\\
108 & 1090$_P$ & 1175$_g$&194 & 2680$_P$ & 2803$_i$&276 & 2760$_M$ & 3575$_g$\\
110 & 1130$_P$ & 1199$_g$&196 & 1176$_M$ & 1365$_g$&278 & 4574$_M$ & 4767$_i$\\
114 & 1192$_P$ & 1277$_g$&198 & 2786$_P$ & 2870$_g$&280 & 1960$_M$ & 2511$_g$\\
116 & 696$_M$ & 805$_g$&200 & 2842$_P$ & 2867$_g$&282 & 4684$_M$ & 4863$_i$\\
118 & 708$_M$ & 936$_k$&202 & 1212$_M$ & 1407$_i$&284 & 4706$_P$ & 4916$_i$\\
122 & 732$_M$ & 1452$_a$&204 & 1224$_M$ & 1421$_i$&286 & 1716$_M$ & 3420$_g$\\
126 & 756$_M$ & 1221$_a$&206 & 1236$_M$ & 1640$_g$&290 & 1740$_M$ & 5202$_a$\\
129 & 903$_M$ & 1472$_a$&209 & 2299$_M$ & 2912$_g$&294 & 5068$_M$ & 5088$_g$\\
130 & 780$_M$ & 903$_g$&210 & 2100$_M$ & 2299$_g$&&& \\
\hline
\end{tabular}
\caption{New lower bounds for $M(n,n-1), n<300$. 
Subscript legend: 
\textit{\textbf{M}} - previous result from MOLS;
\textit{\textbf{P}} - previous result from simple partition and extension \cite{bms17};
\textit{\textbf{a}} - methods described in \cite{manuscript}; 
\textit{\textbf{g}} - partition of positions $\P$ from greedy partition selection algorithm (See Section \ref{ss:greedy_partition_sel}); 
\textit{\textbf{i}} - partition of positions $\P$ from ILP partition selection algorithm (see Section \ref{ss:ilp_partition_sel});  
\textit{\textbf{k}} - PA $\M$ from modified Kronecker product (see Section \ref{s:kron}).
}
\label{tb:n_1a}
\end{table}

\begin{table}[htb]
 \centering
\vspace*{2mm}
\begin{tabular}{|r r r | r r r | r r r|}
\hline
$n$ & \bf Prev & \bf New & 
$n$ & \bf Prev & \bf New & 
$n$ & \bf Prev & \bf New \\
\hline\hline
300 & 2100$_M$ & 3588$_k$&406 & 2842$_M$ & 3240$_k$&494 & 2964$_M$ & 7888$_k$\\
306 & 1836$_M$ & 4575$_i$&408 & 4070$_M$ & 6105$_i$&498 & 2988$_M$ & 7455$_k$\\
308 & 5360$_M$ & 5524$_i$&410 & 2870$_M$ & 8389$_i$&500 & 3500$_M$ & 11373$_i$\\
312 & 5436$_M$ & 5660$_i$&412 & 3296$_M$ & 5343$_g$&504 & 3527$_M$ & 11416$_i$\\
314 & 2198$_M$ & 5723$_i$&414 & 4140$_M$ & 4956$_g$&506 & 3036$_M$ & 7575$_i$\\
316 & 2212$_M$ & 3150$_g$&415 & 3735$_M$ & 4140$_g$&508 & 3556$_M$ & 7605$_i$\\
318 & 2226$_M$ & 5793$_g$&417 & 6255$_M$ & 7481$_i$&510 & 3060$_M$ & 11661$_i$\\
322 & 1932$_M$ & 4815$_g$&418 & 2926$_M$ & 6255$_i$&513 & 9234$_M$ & 11264$_a$\\
324 & 2592$_M$ & 5168$_k$&420 & 2940$_M$ & 8744$_i$&516 & 4128$_M$ & 7725$_g$\\
326 & 1956$_M$ & 3900$_k$&422 & 2954$_M$ & 8822$_i$&518 & 5170$_M$ & 6204$_g$\\
330 & 1980$_M$ & 2961$_g$&424 & 3384$_M$ & 6345$_i$&520 & 4160$_M$ & 7785$_g$\\
332 & 2324$_M$ & 6105$_i$&426 & 2556$_M$ & 6800$_k$&522 & 5220$_M$ & 11983$_i$\\
334 & 2338$_M$ & 2664$_k$&430 & 2580$_M$ & 3003$_g$&524 & 6288$_M$ & 12029$_i$\\
335 & 2010$_M$ & 2338$_g$&432 & 6480$_M$ & 9051$_i$&526 & 4208$_M$ & 7875$_g$\\
338 & 2028$_M$ & 6349$_i$&434 & 2608$_M$ & 9093$_i$&528 & 7920$_M$ & 8432$_k$\\
340 & 2040$_M$ & 2373$_g$&436 & 2616$_M$ & 6525$_i$&530 & 3710$_M$ & 12696$_a$\\
344 & 2408$_M$ & 6076$_a$&438 & 3066$_M$ & 7866$_k$&532 & 4256$_M$ & 7965$_i$\\
346 & 2076$_M$ & 2415$_g$&440 & 3159$_M$ & 9219$_i$&534 & 3738$_M$ & 6396$_k$\\
348 & 2088$_M$ & 6658$_i$&442 & 3528$_M$ & 6615$_i$&536 & 4288$_M$ & 8025$_i$\\
350 & 2800$_M$ & 6714$_i$&444 & 3108$_M$ & 9069$_g$&538 & 5380$_M$ & 8055$_i$\\
354 & 2124$_M$ & 6746$_g$&446 & 3122$_M$ & 5785$_i$&540 & 6480$_M$ & 8085$_k$\\
356 & 2492$_M$ & 3195$_g$&450 & 3220$_M$ & 9429$_g$&542 & 3794$_M$ & 12443$_i$\\
358 & 2148$_M$ & 3213$_g$&452 & 4510$_M$ & 6765$_i$&545 & 8704$_M$ & 9792$_k$\\
360 & 2520$_M$ & 6965$_i$&456 & 3192$_M$ & 6825$_i$&548 & 3836$_M$ & 12581$_i$\\
362 & 2172$_M$ & 7220$_a$&458 & 3206$_M$ & 9644$_g$&550 & 3850$_M$ & 4392$_k$\\
366 & 2196$_M$ & 2555$_g$&460 & 3220$_M$ & 7334$_k$&552 & 5220$_M$ & 9918$_k$\\
368 & 5520$_M$ & 7108$_g$&462 & 3234$_M$ & 10061$_i$&558 & 3906$_M$ & 13329$_i$\\
370 & 2952$_M$ & 5535$_i$&464 & 6960$_M$ & 10162$_i$&561 & 3927$_M$ & 8400$_i$\\
372 & 2604$_M$ & 5565$_i$&466 & 3262$_M$ & 6975$_i$&564 & 3948$_M$ & 13500$_i$\\
374 & 2618$_M$ & 7381$_i$&468 & 3744$_M$ & 10253$_i$&566 & 3396$_M$ & 3955$_g$\\
376 & 2632$_M$ & 5625$_i$&470 & 3290$_M$ & 3752$_g$&570 & 3420$_M$ & 13654$_i$\\
378 & 4524$_M$ & 4901$_i$&472 & 3304$_M$ & 7065$_i$&572 & 4004$_M$ & 13699$_i$\\
380 & 2660$_M$ & 7556$_i$&474 & 4740$_M$ & 7095$_k$&576 & 4608$_M$ & 12650$_k$\\
382 & 2674$_M$ & 4572$_i$&476 & 3332$_M$ & 8550$_k$&578 & 4046$_M$ & 13848$_i$\\
384 & 5760$_M$ & 7692$_i$&478 & 3816$_M$ & 7155$_i$&582 & 4074$_M$ & 4648$_g$\\
386 & 2702$_M$ & 5775$_i$&480 & 7200$_M$ & 10538$_i$&584 & 4088$_M$ & 5830$_i$\\
388 & 3096$_M$ & 5805$_i$&482 & 5772$_M$ & 7215$_i$&586 & 4102$_M$ & 4680$_g$\\
390 & 2730$_M$ & 7897$_i$&484 & 3872$_M$ & 7245$_i$&588 & 4116$_M$ & 14088$_i$\\
392 & 2744$_M$ & 6256$_k$&485 & 3395$_M$ & 3872$_g$&590 & 10030$_M$ & 10602$_k$\\
398 & 2786$_M$ & 7940$_i$&486 & 2916$_M$ & 3395$_g$&591 & 4137$_M$ & 10030$_i$\\
402 & 2814$_M$ & 8020$_i$&488 & 3416$_M$ & 10714$_i$&594 & 4752$_M$ & 14232$_i$\\
404 & 4836$_M$ & 6045$_k$&490 & 2940$_M$ & 7335$_g$&596 & 4172$_M$ & 8925$_i$\\
405 & 3240$_M$ & 4444$_g$&492 & 2952$_M$ & 10802$_i$&600 & 8400$_M$ & 14828$_i$\\
\hline
\end{tabular}
\caption{New lower bounds for $M(n,n-1), (300\le n\le 600$).
Refer to Table \ref{tb:n_1a} for an explanation of the subscripts.
}
\label{tb:n_1b}
\end{table}

\clearpage

\section{Conclusion} \label{s:conclusion}
We have presented new computational methods for the partition and extension technique that produce several competitive new lower bounds on $M(n,d)$ for various integers $n$ and $d$. 
We described sequential partition and extension, which is very useful for improving lower bounds.
The techniques of rudimentary and general parallel partition and extension introduce several new symbols simultaneously. 
They are different extension strategies that provide many improved lower bounds for $M(n,d)$. 
We have given several new techniques and experimental results that provide new lower bounds for $M(n,n-1)$, for many integers $n<600$.

\section*{Acknowledgement}

We would like to thank Zachary Hancock and Alexander Wong, who separately wrote programs to compute some of our improved lower bounds.

\end{document}